\documentclass[submission,copyright,creativecommons]{eptcs}
\providecommand{\event}{ICE 2014} 

\usepackage{amsfonts}
\usepackage{amsmath}
\usepackage{amssymb}
\usepackage{amsbsy}
\usepackage{enumerate}
\usepackage{stackrel}
\usepackage{bm}
\usepackage[toc,page]{appendix}

\usepackage{pifont}
\newcommand{\cmark}{\mbox{\tt ok}}%
\newcommand{\xmark}{\mbox{\tt no}}%

\newdimen\proofrulebreadth \proofrulebreadth=.05em
\newdimen\proofdotseparation \proofdotseparation=1.25ex
\newdimen\proofrulebaseline \proofrulebaseline=2ex
\newcount\proofdotnumber \proofdotnumber=3
\let\then\relax
\def\hfi{\hskip0pt plus.0001fil}
\mathchardef\squigto="3A3B
\newif\ifinsideprooftree\insideprooftreefalse
\newif\ifonleftofproofrule\onleftofproofrulefalse
\newif\ifproofdots\proofdotsfalse
\newif\ifdoubleproof\doubleprooffalse
\let\wereinproofbit\relax
\newdimen\shortenproofleft
\newdimen\shortenproofright
\newdimen\proofbelowshift
\newbox\proofabove
\newbox\proofbelow
\newbox\proofrulename
\def\shiftproofbelow{\let\next\relax\afterassignment\setshiftproofbelow\dimen0 }
\def\shiftproofbelowneg{\def\next{\multiply\dimen0 by-1 }%
\afterassignment\setshiftproofbelow\dimen0 }
\def\setshiftproofbelow{\next\proofbelowshift=\dimen0 }
\def\setproofrulebreadth{\proofrulebreadth}

\def\prooftree{
\ifnum  \lastpenalty=1
\then   \unpenalty
\else   \onleftofproofrulefalse
\fi
\ifonleftofproofrule
\else   \ifinsideprooftree
        \then   \hskip.5em plus1fil
        \fi
\fi
\bgroup
\setbox\proofbelow=\hbox{}\setbox\proofrulename=\hbox{}%
\let\justifies\proofover\let\leadsto\proofoverdots\let\Justifies\proofoverdbl
\let\using\proofusing\let\[\prooftree
\ifinsideprooftree\let\]\endprooftree\fi
\proofdotsfalse\doubleprooffalse
\let\thickness\setproofrulebreadth
\let\shiftright\shiftproofbelow \let\shift\shiftproofbelow
\let\shiftleft\shiftproofbelowneg
\let\ifwasinsideprooftree\ifinsideprooftree
\insideprooftreetrue
\setbox\proofabove=\hbox\bgroup$\displaystyle 
\let\wereinproofbit\prooftree
\shortenproofleft=0pt \shortenproofright=0pt \proofbelowshift=0pt
\onleftofproofruletrue\penalty1
}

\def\eproofbit{
\ifx    \wereinproofbit\prooftree
\then   \ifcase \lastpenalty
        \then   \shortenproofright=0pt  
        \or     \unpenalty\hfil         
        \or     \unpenalty\unskip       
        \else   \shortenproofright=0pt  
        \fi
\fi
\global\dimen0=\shortenproofleft
\global\dimen1=\shortenproofright
\global\dimen2=\proofrulebreadth
\global\dimen3=\proofbelowshift
\global\dimen4=\proofdotseparation
\global\count255=\proofdotnumber
$\egroup  
\shortenproofleft=\dimen0
\shortenproofright=\dimen1
\proofrulebreadth=\dimen2
\proofbelowshift=\dimen3
\proofdotseparation=\dimen4
\proofdotnumber=\count255
}

\def\proofover{
\eproofbit 
\setbox\proofbelow=\hbox\bgroup 
\let\wereinproofbit\proofover
$\displaystyle
}%
\def\proofoverdbl{
\eproofbit 
\doubleprooftrue
\setbox\proofbelow=\hbox\bgroup 
\let\wereinproofbit\proofoverdbl
$\displaystyle
}%
\def\proofoverdots{
\eproofbit 
\proofdotstrue
\setbox\proofbelow=\hbox\bgroup 
\let\wereinproofbit\proofoverdots
$\displaystyle
}%
\def\proofusing{
\eproofbit 
\setbox\proofrulename=\hbox\bgroup 
\let\wereinproofbit\proofusing
\kern0.3em$
}

\def\endprooftree{
\eproofbit 
  \dimen5 =0pt
\dimen0=\wd\proofabove \advance\dimen0-\shortenproofleft
\advance\dimen0-\shortenproofright
\dimen1=.5\dimen0 \advance\dimen1-.5\wd\proofbelow
\dimen4=\dimen1
\advance\dimen1\proofbelowshift \advance\dimen4-\proofbelowshift
\ifdim  \dimen1<0pt
\then   \advance\shortenproofleft\dimen1
        \advance\dimen0-\dimen1
        \dimen1=0pt
        \ifdim  \shortenproofleft<0pt
        \then   \setbox\proofabove=\hbox{%
                        \kern-\shortenproofleft\unhbox\proofabove}%
                \shortenproofleft=0pt
        \fi
\fi
\ifdim  \dimen4<0pt
\then   \advance\shortenproofright\dimen4
        \advance\dimen0-\dimen4
        \dimen4=0pt
\fi
\ifdim  \shortenproofright<\wd\proofrulename
\then   \shortenproofright=\wd\proofrulename
\fi
\dimen2=\shortenproofleft \advance\dimen2 by\dimen1
\dimen3=\shortenproofright\advance\dimen3 by\dimen4
\ifproofdots
\then
        \dimen6=\shortenproofleft \advance\dimen6 .5\dimen0
        \setbox1=\vbox to\proofdotseparation{\vss\hbox{$\cdot$}\vss}%
        \setbox0=\hbox{%
                \advance\dimen6-.5\wd1
                \kern\dimen6
                $\vcenter to\proofdotnumber\proofdotseparation
                        {\leaders\box1\vfill}$%
                \unhbox\proofrulename}%
\else   \dimen6=\fontdimen22\the\textfont2 
        \dimen7=\dimen6
        \advance\dimen6by.5\proofrulebreadth
        \advance\dimen7by-.5\proofrulebreadth
        \setbox0=\hbox{%
                \kern\shortenproofleft
                \ifdoubleproof
                \then   \hbox to\dimen0{%
                        $\mathsurround0pt\mathord=\mkern-6mu%
                        \cleaders\hbox{$\mkern-2mu=\mkern-2mu$}\hfill
                        \mkern-6mu\mathord=$}%
                \else   \vrule height\dimen6 depth-\dimen7 width\dimen0
                \fi
                \unhbox\proofrulename}%
        \ht0=\dimen6 \dp0=-\dimen7
\fi
\let\doll\relax
\ifwasinsideprooftree
\then   \let\VBOX\vbox
\else   \ifmmode\else$\let\doll=$\fi
        \let\VBOX\vcenter
\fi
\VBOX   {\baselineskip\proofrulebaseline \lineskip.2ex
        \expandafter\lineskiplimit\ifproofdots0ex\else-0.6ex\fi
        \hbox   spread\dimen5   {\hfi\unhbox\proofabove\hfi}%
        \hbox{\box0}%
        \hbox   {\kern\dimen2 \box\proofbelow}}\doll%
\global\dimen2=\dimen2
\global\dimen3=\dimen3
\egroup 
\ifonleftofproofrule
\then   \shortenproofleft=\dimen2
\fi
\shortenproofright=\dimen3
\onleftofproofrulefalse
\ifinsideprooftree
\then   \hskip.5em plus 1fil \penalty2
\fi
}

\title{Loosening the notions of compliance and sub-behaviour in client/server systems}

\author{Franco Barbanera
\institute{Dipartimento di Matematica e Informatica\\
University of Catania}
\email{barba@dmi.unict.it}
 \and 
Ugo de'Liguoro
\institute{Dipartimento di Informatica\\
University of Torino}
\email{ugo.deliguoro@di.unito.it}
}

\newtheorem{definition}{Definition}[section]
\newtheorem{lemma}[definition]{Lemma}
\newtheorem{proposition}[definition]{Proposition}
\newtheorem{theorem}[definition]{Theorem}
\newtheorem{corollary}[definition]{Corollary}
\newtheorem{remark}[definition]{Remark}
\newtheorem{example}[definition]{Example}

\newtheorem{fact}[definition]{Fact}

\DeclareMathAlphabet{\mathpzc}{OT1}{pzc}{m}{it}

\newcommand{\qde}{\hfill $\Box$}

\newcommand{\BackRed}[1]{\stackrel{\mbox{\vspace{-8pt}\normalsize $\looparrowright$}}{\looparrowright}}

\renewcommand{\vec}[1]{\bm{#1}}

\newenvironment{proof}{{\em Proof.~}}{~~\qde \medskip}

\newcommand{\Comment}[1]{ }

\newcommand{\ByDef}{\triangleq}

\newcommand{\Rel}{\mathpzc R}
\newcommand{\RelKK}{\mathpzc K}

\newcommand{\FunH}{{\cal H}}

\newcommand{\FunJ}{{\cal J}}

\newcommand{\false}{{\sf false}}

\newcommand{\Subst}[2]{\{#1/#2\}}

\newcommand{\fv}[1]{\mbox{\sc fv}(#1)}

\newcommand{\Lts}[1]{\stackrel{#1}{\Longrightarrow}}

\newcommand{\comply}{\dashv}
\newcommand{\complyP}{\dashv^{\mbox{\tiny {\sf P}}}}

\newcommand{\complyG}{\comply^{\mbox{\tiny {\tt skp}}}}
\newcommand{\complyGco}{{\comply}^{\mbox{\tiny {\tt skp}}}_{co}}
\newcommand{\complyGcok}[1]{{\comply}^{\mbox{\tiny {\tt skp}}}_{co.#1}}

\newcommand{\csAct}{\mbox{\bf sAct}}

\newcommand{\skipAct}{{\tt skp}}
\newcommand{\orchAct}[2]{\langle\mbox{\small $#1$},\mbox{\small $#2$}\rangle}

\newcommand{\Dual}[1]{\overline{#1}}
\newcommand{\aDual}[1]{\tilde{#1}}

\newcommand{\der}{\vartriangleright}

\newcommand{\ScomplUnfoldL}{\mbox{\sc Unf-L}}
\newcommand{\ScomplUnfoldR}{\mbox{\sc Unf-R}}
\newcommand{\ScomplUnfoldLR}{\mbox{\sc Unf-L-R}}
\newcommand{\ScomplHyp}{\mbox{\sc Hyp}}
\newcommand{\ScomplAx}{\mbox{\sc Ax}}
\newcommand{\ScomplSumOplus}{\mbox{\sc $+$.$\oplus$-Cpl}}
\newcommand{\ScomplOplusOplus}{\mbox{\sc $\oplus$.$\oplus$-Cpl}}
\newcommand{\ScomplOplusSum}{\mbox{\sc $\oplus$.$+$-Cpl}}

\newcommand{\Nat}{\mathbb{N}} 

\newcommand{\Set}[1]{\{#1\}}

\newcommand{\Iff}{\Leftrightarrow}

\newcommand{\Or}{\;\vee\;}

\renewcommand{\implies}{\Rightarrow}

\newcommand{\lts}[1]{\stackrel{#1}{\longrightarrow}}

\newcommand{\LtsM}[1]{\stackrel[max]{#1}{\Longrightarrow}}
\newcommand{\rec}{{\sf rec} \, }

\newcommand{\substr}{\sqsubseteq}

\newcommand{\peer}{\ast}

\newcommand{\Names}{{\cal N}}
\newcommand{\CoNames}{\overline{\Names}}

\newcommand{\stopA}{{\bf 1}}
\newcommand{\Act}{\mbox{\bf Act}}

\newcommand{\noLts}[2]{#1 \not\Downarrow #2}

\newcommand{\preceqG}{\preceq^{\mbox{\tiny {\tt skp}}}}
\newcommand{\preceqP}{\preceq^{\mbox{\tiny {\sf P}}}}

\newcommand{\preceqBdL}{\preceq^{\mbox{\tiny {\sf BdL}}}}
\newcommand{\preceqGH}{\preceq^{\mbox{\tiny {\sf GH}}}}

\newcommand{\Rbehav}{{\sf BE}}
\newcommand{\Sbehav}{{\sf SB}}

\newcommand{\IntChoice}{\oplus}
\newcommand{\unfoldAction}{\rec}

\newcommand{\Converge}[2]{#1\!\Downarrow#2}
\newcommand{\notConverge}[2]{#1\!\not\Downarrow#2}

\newcommand{\myshrink}{}

\newcommand{\trace}{\mbox{\sf Tr}}
\newcommand{\cstrace}{\mbox{\sf sTr}}

\newcommand{\sync}[2]{#1\; \mbox{\sf synch}\, #2}

\def\titlerunning{Loosening compliance and sub-behaviour}
\def\authorrunning{F. Barbanera \& U. de'Liguoro}
\begin{document}

\maketitle

\begin{abstract} In the context of ``session behaviors'' for client/server systems,
we propose a weakening of the compliance and sub-behaviour relations where the bias toward the client (whose ``requests'' must be satisfied) 
is pushed further with respect to the usual definitions, by admitting that ``not needed'' output actions from the server side can be {\em skipped} 
by the client.
Both compliance and sub-behaviour relations resulting from this weakening remain
decidable, though the proof of the duals-as-minima property for servers, on which the decidability of the sub-behaviour relation relies, 
requires a tighter analysis of client/server interactions.

%
%

\end{abstract}




\section{Introduction}

The formal specification of web-services behaviour is a crucial issue toward automatic discovery and
composition of software modules available through a network. Among several approaches we consider here the theory of contracts introduced in \cite{CCLP06} and developed in a series of papers e.g. \cite{CP09,CGP08}. We focus here on the scenario of client/server architecture, where services stored in a repository are queried by clients to establish a two-sided communication. 

To check the matching of client's requirements against the service offered by a server, both server and client behaviours are described via a CCS-like formalism (without $\tau$-actions nor parallel composition), whose terms are dubbed {\em contracts}. 
The basic notion studied in the theory is the {\em compliance} relation, written $\rho\comply\sigma$, meaning that all requirements by the client $\rho$ are eventually matched by some communication action by the server $\sigma$
\footnote{It is not feasible, however, to allow the client to terminate the interaction at any point,
since, trivially, any server would be compliant with such a sort of client.}. This is mathematically defined using an LTS semantics of the communication behaviour of the pair of contracts $\rho\|\sigma$, where $\rho\|\sigma\lts{}\rho'\|\sigma'$ holds  whenever $\rho\lts{\alpha}\rho'$, $\sigma\lts{\Dual{\alpha}}\sigma'$ and
$\alpha$ and $\Dual{\alpha}$ are dual actions.
Now, writing $\Lts{}$ for the reflexive and transitive closure of $\lts{}$, the relation $\rho\comply\sigma$ holds if and only if $\rho\|\sigma \Lts{} \rho'\|\sigma' \not\!\!\Lts{}$ implies $\rho' = \stopA$, where $\stopA$ is the behaviour of the completed
process. When $\rho\comply\sigma$ we say that $\rho$ is a {\em client} of the {\em server} $\sigma$, slightly abusing terminology.

The compliance relation characterises client/server interaction with a bias toward the client, which is the sole guaranteed to complete. To illustrate this by an example, let us consider a ballot service whose behaviour is described by the following server contract:
\[{\sf BallotServiceAB} \ByDef \rec x.\;{\tt Login}.(\Dual{\tt Wrong}.x \;\oplus\; \Dual{\tt Overload}.x \;\oplus\;
                              \Dual{\tt Ok}.({\tt VoteA}+{\tt VoteB}) ).
\]
This service can receive a login from a client, a voter, via the input action {\tt Login}; if the login is correct the server issues to the client the message $\Dual{\tt Ok}$ (an output action), enabling the client to vote for either candidates A or  B via a continuation consisting of the external choice $+$ of the input actions {\tt VoteA} and {\tt VoteB}. 
In case the login is incorrect or the service is busy, the messages $\Dual{\tt Wrong}$ and $\Dual{\tt Overload}$ are sent to the client respectively, both by output actions. In both cases the voter is allowed to retry the login by recursion. The output actions
$\Dual{\tt Ok}$, $\Dual{\tt Wrong}$ and $\Dual{\tt Overload}$ are composed by an internal choice $\oplus$ since they depend on internal decisions on the server side. 
Now let us  consider the following client:
\[{\sf Voter}  \ByDef \rec x.\; \Dual{\tt Login}.({\tt Wrong}.x + {\tt Overload}.x + {\tt Ok}.\Dual{\tt VoteA}).\]
{\sf Voter} will not give up synchronizing with {\sf BallotService} until eventually allowed to send her vote. 
According to the definition of compliance we have that ${\sf Voter}\comply {\sf BallotServiceAB}$, and this remains true also in the case of the slightly
different server:
\[{\sf BallotServiceABC} \ByDef \rec x.\;{\tt Login}.(\Dual{\tt Wrong}.x \;\oplus\;
                              \Dual{\tt Ok}.({\tt VoteA}+{\tt VoteB} + {\tt VoteC}) ).
\]
which is not willing to issue the $\Dual{\tt Overload}$ message, and allows one more candidate to be voted. Indeed what matters is the fact that no interaction among client and server will ever get stuck in a state in which some client action is pending. 
Because of the same reason the client {\sf Voter} is also compliant with the service:
\[\rec x.\;{\tt Login}.(\Dual{\tt Wrong}.x \;\oplus\; \Dual{\tt Overload}.x \;\oplus\;
                              \Dual{\tt Ok}.({\tt VoteA}.({\tt Va1} + {\tt Va2})+{\tt VoteB}.({\tt Vb1} + {\tt Vb2})) ),
\]
where {\tt Va1} and {\tt Va2} are choices depending on the vote {\tt VoteA}, and similarly for {\tt Vb1} and {\tt Vb2}. However {\sf Voter} is not compliant
with
\begin{tabbing}
${\sf Ball} \ByDef$ \= $\rec x.\;{\tt Login}.(\Dual{\tt Wrong}.\Dual{\tt InfoW}.x\;\oplus\; \Dual{\tt Overload}.x \;\oplus\; \Dual{\tt Ok}.\Dual{\tt Id}.($\= \kill

${\sf BallotServiceBehSkp} \ByDef $\\
   \> $\rec x.\;{\tt Login}.(\Dual{\tt Wrong}.\Dual{\tt InfoW}.x\;\oplus\; \Dual{\tt Overload}.x \;\oplus\; \Dual{\tt Ok}.\Dual{\tt Id}.(\;\;\;\;{\tt VoteA}.({\tt Va1} + {\tt Va2})$\\
             \>                                  \>$+\; {\tt VoteB}.({\tt Vb1} + {\tt Vb2})\;)\;)$ \\
\end{tabbing}
because of the actions $\Dual{\tt InfoW}$ and $\Dual{\tt Id}$ (the former representing   infos about the failure of the login and the latter representing an identifier of the voting transaction), that do not have any correspondent input on the client side. However these outputs
have hardly any control significance, which is especially the case in the setting of {\em session-behaviours} we have introduced in \cite{BdL10} and that have been also investigated in \cite{BH13} (where they are dubbed {\em session contracts}). In fact session-behaviours are contracts in which the only terms that can occur in an internal choice have to be prefixed by the output of pairwise distinct messages (the internal choice being the only truly non-deterministic feature of a session-behaviour).

In this paper we investigate the possibility of loosening the notion of compliance for session behaviours by admitting that a client, before an actual syncronization, can {\em skip} (disregard) a finite number of consecutive output actions by the server, provided that these are not the dual of some immediate input actions of the client. 
The overall number of (non consecutive) skipped actions in an interaction, however, can be possibly infinite. 
We call the resulting relation $\skipAct$-{\em compliance} and write $\rho\complyG\sigma$ for ``$\rho$ is $\skipAct$-compliant with $\sigma$''. There is a contrast between these two conditions; while the latter is easily decidable by looking at the contract syntax (and by admitting only guarded recursion), the former is an infinitary condition, ruling out those infinite interactions which happen to be {\em definitely} skip actions. The first result which we obtain is that, in spite of its infinitary definition, the so obtained compliance notion is decidable.

Compliance naturally induces a preorder over contracts seen as the behavioural specification of a server. In \cite{BdL10,BdL13} we say that $\sigma\preceq_s\sigma'$  if any client of $\sigma$ is also a client of $\sigma'$ according to the compliance relation $\comply$. It can be checked that, for example, 
${\sf BallotServiceAB}\preceq_s{\sf BallotServiceABC}$, but neither of them is comparable to ${\sf BallotServiceBehSkp}$. By replacing $\complyG$ in this definition one obtains a similar preorder $\sigma\preceqG\sigma'$, which also turns out to be decidable. The proof of the latter fact relies on the notion of dual behaviour $\Dual{\rho}$ of $\rho$ and on the property that $\Dual{\rho}$ is the minimal server of $\rho$ w.r.t. $\preceqG$.\\

\medskip
\noindent
{\em Overview of the paper.} The notion of session-behaviour is recalled in Section \ref{sect:skipCompliance}. 
Then the definition of $\skipAct$-compliance is given in Subsection \ref{subsect:skipCompliance}.
In Section \ref{sect:coChar} it is provided a coinductive characterization of 
$\skipAct$-compliance, via a formal system to deduce (conditional) $\skipAct$-compliance, which is proved to be sound and complete. 
Decidability then follows, being the system algorithmic.
The notion of $\skipAct$-subbehaviour $\preceqG$ is introduced in Section \ref{sect:skipSub}, and 
the property of {\em duals as minima} is proved. Decidability of $\preceqG$ is a consequence
 of such a property.
In Section \ref{sect:skipRelated} we extensively discuss the relationship of our $\skipAct$-compliance
with another weak notion of compliance allowing for a sort of ``action skipping'': the (orchestrated) weak-compliance proposed by Padovani in 
\cite{Padovani10}. 
In Section \ref{sect:futureWork} a discussion about  future works  concludes the paper.

\if false 

\newpage
Compliance naturally induces a preorder $\preceqBdL_s$ over contracts called {\em server sub-behaviour} in  \cite{BdL10} (in contrast with {\em client sub-behaviour}), 
where $\sigma \preceqBdL_s \tau$ holds whenever all clients of $\sigma$ are clients of $\tau$.

In \cite{BdL10} the formalism of {\em session-behaviours} is introduced (a restriction of the more
general notion of contracts) in order to represent the possible interactions offered by components  in session-based
client/server distributed systems. 
The relations $\preceqBdL_c$, $\preceqBdL_s$ and  $\preceqBdL_\peer$ are then introduced and thoroughly investigated in \cite{BdL13}, representing the notion of substitutability between, respectively, clients, servers and peers in client/server-based distributed systems taking into account the different roles a component plays in an interaction (the superscript {\sf BdL} is here used to help to clearly distinguishing these relations from the other treated in this paper.)

These sub-behaviour relations are the natural preorders on session behaviours induced by the relation $\comply$ of {\em compliance}, introduced in  \cite{LP07,PadTCS} on contracts. The  relation $\rho\comply\sigma$ ensures that any request from the client $\rho$ is satisfied by the server $\sigma$, so that any possible interaction among $\rho$ and $\sigma$ will never prevent the client from completing; this can be seen as the success condition of testing $\sigma$ against $\rho$ (client/server/peer relations for first-order contracts  has been investigated by Hennessy et al. in \cite{BH13c}.) In \cite{BdL13} on the one hand the notion of compliance is restricted to session-behaviour and on the other hand higher-order behaviours are also taken into account. 

The restricted sintax of session-behaviours enables them to be looked at from a process
calculi point of view, as well all from a type-theory point of view, being them in one-to-one
correspondence with session types as defined in \cite{GH05}. As shown in \cite{BdL13},
the relation $\preceqBdL_s$ can be looked at as an extension of the Gay-Hole subtype relation
of  \cite{GH05} in which the server is not committed to complete in case the client does
(the other two relations $\preceqBdL_c$ and  $\preceqBdL_\peer$ represent type substitutability
for clients and "peers").

In order to intuitively recall such notions, we present a working example 
sketched in \cite{BdL10} which, in turn,
is an adaptation of one in \cite{LaneveP08}.
We shall then extend this example in order to provide intuition for the notions introduced
in the present paper.\\

Let us consider a {\em Ballot-Service}.
This service can receive a login and, if correct, it signals to the client (a voter), by means of the message {\tt Ok}, that it is enabled
to vote either for the candidate A or for the canditate  B. 
In case the login is incorrect, instead, a message {\tt Wrong} is issued to the client.
By means of recursion, a voter is allowed to retry the login action in case of a failure.

\myshrink
\begin{tabbing}
${\sf BallotServiceBeh} \ByDef$ \= $\rec x.\;{\tt Login}.($\=$\Dual{\tt Wrong}.x \;\oplus\;$ 
                              $\Dual{\tt Ok}.($\=${\tt VoteA}+{\tt VoteB}) )$
\end{tabbing}
\myshrink

If one takes into account a notion of compliance requiring the completion of both the 
client and the server, we get a (server) sub-behaviour relation (let us dub it $\preceqGH_s$) corresponding to the
Gay-Hole subtyping relation. In fact, by defining 

\begin{tabbing}
${\sf BallotSer} \ByDef$ \= $\rec x.\;{\tt Login}.($\=$\Dual{\tt Ok}.($\= \kill

${\sf BallotServiceBehGH} \ByDef $\\
  \> $\rec x.\;{\tt Login}.(\Dual{\tt Wrong}.~x$ $\;\oplus\; $ $\Dual{\tt Overload}.~x $  $\;\oplus\; $ $\Dual{\tt Ok}.($\=${\tt VoteA}+{\tt VoteB}+{\tt VoteC})\;)$
\end{tabbing} 
\Comment{
\begin{tabbing}
${\sf BallotServiceBehGH} \ByDef$ \= $\rec x.\;{\tt Login}.($\=$\Dual{\tt Ok}.($\= \kill

${\sf BallotServiceBehGH} \ByDef \rec x.\;{\tt Login}.(\Dual{\tt Wrong}.~x$\\
                                        \>                  \> $\;\oplus\; $\\
                                        \>  \> $\Dual{\tt Overload}.~x $ \\
                                         \> \> $\;\oplus\; $\\
                                         \>\> $\Dual{\tt Ok}.($\=${\tt VoteA}+{\tt VoteB}+{\tt VoteC})$
\end{tabbing} 
where $\Dual{\tt Overload}$ represent a message concerning the impossibility voting due to a
current work overload of the server, we get: ~~
${\sf BallotServiceBeh}\preceqGH_s {\sf BallotServiceBehGH}$.\\
}

Taking into account all the possible roles of the interacting components in a system (client/server/peer), a server does not have necessarily to ''complete'' (i.e. to reduce to $\stopA$) in case its client does;
what is actually required is that any client's "request" be satisfied by the server.
This requirement can be formalized in many different ways, each corresponding
to particular implementation of the client/server interactions in a system.

In \cite{BdL13} and in \cite{PadTCS} (and in many other papers) this is achieved by means of an operational semantics described through an LTS defining a reduction relation among parallel composition of contracts,  such that $\rho\|\sigma \Lts{} \rho'\|\sigma'$ if $\rho \Lts{\alpha} \rho'$ and $\sigma \Lts{\aDual{\alpha}} \sigma'$, and $\alpha$ and $\aDual{\alpha}$ are actions that {\em synchronise} with each other. The compliance relation $\rho\comply\sigma$ is then defined by the clause that, if $\rho\|\sigma \Lts{} \rho'\|\sigma' \not\Lts{}$, then $\rho' = \stopA$, formally describing the fact that the ''server''
$\sigma$ has satisfied all the requests of the client $\rho$.\\
Informally, the above LTS corresponds to a sort of depth subtyping, where a client can
decide in any moment to {\em abort} the interaction offered  by the server.

In our ballot context this formalization of the notion of compliance makes the following hold:
${\sf BallotServiceBeh}\preceqBdL_s {\sf BallotServiceBehBdL}$, 
where

\begin{tabbing}
${\sf B} \ByDef$ \= $ \rec x.\;{\tt Login}.(\Dual{\tt Wrong}.x$ $\;\oplus\; $ $\Dual{\tt Overload}.~x $  $\;\oplus\; $ $\Dual{\tt Ok}.($\=${\tt VoteA}.($ \=${\tt Va1})$ $+$ ${\tt VoteB}.($\= ${\tt Vb1})$ $+$ ${\tt VoteC}.($\= \kill

${\sf BallotServiceBehBdL} \ByDef $\\
\> $ \rec x.\;{\tt Login}.(\Dual{\tt Wrong}.x$ $\;\oplus\; $ $\Dual{\tt Overload}.~x $  $\;\oplus\; $ $\Dual{\tt Ok}.($\=${\tt VoteA}.({\tt Va1})$ $+$ ${\tt VoteB}.({\tt Vb1})$ $+$ ${\tt VoteC}.({\tt Vc1})\; ) \;)$ \\
             \>                \>                  \>$+$   \>  $+$   \>  $+$\\
            \>                 \>                  \> $ {\tt Va2}$  \>   ${\tt Vb2}$ \> ${\tt Vc2}$ 
\end{tabbing} 
\Comment{
\begin{tabbing}
${\sf BallotServiceBehBdL} \ByDef$ \= $\rec x.\;{\tt Login}.($\=$\Dual{\tt Ok}.($\= \kill

${\sf BallotServiceBehBdL} \ByDef \rec x.\;{\tt Login}.(\Dual{\tt Wrong}.x$\\
                                        \>                  \> $\;\oplus\; $\\
                                        \>  \> $\Dual{\tt Overload}.~x $ \\
                                         \> \> $\;\oplus\; $\\
                                         \>\> $\Dual{\tt Ok}.($\=${\tt VoteA}.({\tt Va1} + {\tt Va2})$\\
             \>                \>                  \>$+$\\
             \>                \>                  \>${\tt VoteB}.({\tt Vb1} + {\tt Vb2})$ \\
             \>                \>                  \>$+$\\
             \>                \>                  \>${\tt VoteC}.({\tt Vc1} + {\tt Vc2})) $
                                                     $)$
\end{tabbing} 
}

Notice, however, that the possibility of aborting an interaction on a client's side is a 
far stronger requirement than having any client's request {\em satisfied}.\\

It is  possible to further weakening the sub-behaviour relation, without losing its
correspondence to a feasible actual implementation.

A weak version of the  notion of compliance on (first-order and unrestricted) contracts has been defined and investigated in \cite{Padovani10}.
There the interactions between a client and a server can be mediated (coordinated) by
an {\em orchestrator}, a form of process (a  sort of {\em active channel} or {\em channel controller}) with the capability of buffering messages.
The presence of an orchestrator enables a form of asyncronous interaction resulting in a weaker
version of the compliance relation, where, due to the presence of buffers in the orchestrators, the server's "answers" to the client's  "requests" can be provided in a different order, so extending the set of the servers that can actually comply with a given client. The weak-compliance
relation of \cite{Padovani10} induces a preorder that the author investigates in that paper, and that here we refer to as $\preceqP_s$.\\

Let us take into account the following behaviour of a ballot service, where ${\tt Id}$ denotes 
an identifier of the transaction.

\begin{tabbing}
${\sf BallotSe} \ByDef$ \= $\rec x.\;{\tt Login}.(\Dual{\tt Wrong}.x \;\oplus\; \Dual{\tt Overload}.~x \;\oplus\;\Dual{\tt Ok}.\Dual{\tt Id}.$\= $\;$ \= \kill

${\sf BallotServiceBehP} \ByDef $\\
 \> $\rec x.\;{\tt Login}.(\Dual{\tt Wrong}.x \;\oplus\; \Dual{\tt Overload}.~x \;\oplus\;\Dual{\tt Ok}.\Dual{\tt Id}.(\; \;\;\;({\tt Va1}.{\tt VoteA} + {\tt Va2}.{\tt VoteA})$\\
             \>                \>                  \>$+ \; ({\tt Vb1}.{\tt VoteB} + {\tt Vb2}.{\tt VoteB})$ \\
             \>                \>                  \>$+\; ({\tt Vc1}.{\tt VoteC} + {\tt Vc2}.{\tt VoteC})\;) \;)$
\end{tabbing} 
\Comment{
\begin{tabbing}
${\sf BallotServiceBehP} \ByDef$ \= $\rec x.\;{\tt Login}.($\=$\Dual{\tt Ok}.\Dual{\tt Id}.((\;\;$\= \kill

${\sf BallotServiceBehP} \ByDef \rec x.\;{\tt Login}.(\Dual{\tt Wrong}.x$\\
                                        \>                  \> $\oplus\; $\\
                                        \>  \> $\Dual{\tt Overload}.~x $ \\
                                         \> \> $\oplus\; $\\
                                         \>\> $\Dual{\tt Ok}.\Dual{\tt Id}.($\=$({\tt Va1}.{\tt VoteA} + {\tt Va2}.{\tt VoteA})$\\
             \>                \>                  \>$+$\\
             \>                \>                  \>$({\tt Vb1}.{\tt VoteB} + {\tt Vb2}.{\tt VoteB})$ \\
             \>                \>                  \>$+$\\
             \>                \>                  \>$({\tt Vc1}.{\tt VoteC} + {\tt Vc2}.{\tt VoteC}) $
                                                     $)$
\end{tabbing} 
}

Now, given the following behaviour of a possible voter,
\begin{tabbing}
${\sf VoterBeh} \ByDef  \Dual{\tt Login}.({\tt Wrong} + {\tt Overload} + {\tt Ok}.\Dual{\tt VoteB}.\Dual{\tt Vb1} ) $
\end{tabbing}
the feasibility of the interaction between ${\sf VoterBeh}$ and ${\sf BallotServiceBehP}$ is guaranteed by a the following orchestrator

\begin{tabbing}
$f \;\ByDef\;\langle{\tt Login},\Dual{\tt Login} \rangle.($ \=\kill

$f \;\ByDef \; \langle{\tt Login},\Dual{\tt Login} \rangle.(\;\;\;\;\;\langle\Dual{\tt Wrong},{\tt Wrong} \rangle $\\
       \> $\vee \; \langle\Dual{\tt Overload},{\tt Overload} \rangle $\\
        \> $\vee \; \langle\Dual{\tt Ok},{\tt Ok} \rangle.\langle\varepsilon, {\tt Id} \rangle.\langle{\tt VoteB},\varepsilon \rangle.\langle {\tt Vb1},\Dual{\tt Vb1},\rangle.\langle\varepsilon, \Dual{\tt VoteB}\rangle )$
\end{tabbing} 
\Comment{
\begin{tabbing}
$f \ByDef\langle{\tt Login},\Dual{\tt Login} \rangle.($ \=\kill

$f \ByDef\langle{\tt Login},\Dual{\tt Login} \rangle.(\langle\Dual{\tt Wrong},{\tt Wrong} \rangle $\\
       \> $\vee $\\
       \> $\langle\Dual{\tt Overload},{\tt Overload} \rangle $\\
        \> $\vee $\\
        \> $\langle\Dual{\tt Ok},{\tt Ok} \rangle.\langle\varepsilon, {\tt Id} \rangle.\langle{\tt VoteB},\varepsilon \rangle.\langle {\tt Vb1},\Dual{\tt Vb1},\rangle.\langle\varepsilon, \Dual{\tt VoteB},\rangle )$
\end{tabbing} 
}

When $f$ receives a  $\Dual{\tt Login}$ message from the client, this message is immediately delivered to the server. Then, in case $f$ gets a message $\Dual{\tt Wrong}$, $\Dual{\tt Overload}$ or $\Dual{\tt Ok}$ from the server, this message is delivered to the client.
Moreover, after receiving ${\tt Ok}$, the subsequent message ${\tt VoteB}$ is virtually kept in a
 bounded buffer and delivered to the server only after the message ${\tt Vb1}$ has been exchanged. \\

Hence the presence of an orchestrator allows for asyncronous interactions, thanks to the 
virtual presence of (bound) buffers in orchestrators.\\

\noindent
Notice that the natural restriction imposed on Padovani's orchestrators - for instance, an orchestrator cannot send a message if this has not been previously received - somewhat 
resembles the restriction imposed on the syntax of session-behaviours. 
For instance, we have that $\Dual{a}.b \oplus \Dual{c}.d\not\preceqP_s
b.\Dual{a} \oplus \Dual{c}.d$, since the orchestrator between the possible client of $\Dual{a}.b \oplus \Dual{c}.d$ and the server  $b.\Dual{a} \oplus \Dual{c}.d$ is not allowed to send the message $b$ that has not received. Similarly, it is not possible to take into account $b.\Dual{a} \oplus \Dual{c}.d$ as a server in the session-behaviours formalism, since it is not a sintactically correct session-behaviour.)\\

Now, taking into account Padovani weak-compliance and sub-behaviour, the following relation holds in our example:
$${\sf BallotServiceBeh}\preceqP_s {\sf BallotServiceBehP}$$

The extreme generality of Padovani's notion of compliance through orchestrators
is paid in terms of a formalization of the compliance relation by means of an operational semantics described through an LTS which depends
on an orchestrator $f$: the reduction relation among parallel composition of contracts is  such that $\rho\|_f\sigma \Lts{} \rho'\|_{f'}\sigma'$ whenever $\rho \Lts{\alpha} \rho'$, $\sigma \Lts{\Dual{\alpha}} \sigma'$ and $f$ reduces to $f'$ by means of a $\langle\Dual{\alpha},\alpha\rangle$ action. We have instead 
 $\rho\|_f\sigma \Lts{} \rho\|_{f'}\sigma'$
(resp. $\rho\|_f\sigma \Lts{} \rho'\|_{f'}\sigma$) whenever $\rho \Lts{\alpha} \rho'$ (resp. $\sigma \Lts{\Dual{\alpha}} \sigma'$) and f reduces to $f'$ by means of a $\langle\Dual{\alpha},\varepsilon\rangle$ (resp. $\langle\varepsilon,\alpha\rangle$) action. We shall refer to such a compliance relation as $\complyP_f$, writing $\rho \complyP \sigma$ when there exists $f$ such that $\rho \complyP_f \sigma$.

As shown by Padovani in \cite{Padovani10}, the corresponding sub-behaviour relation $\sigma\preceqP_s\sigma'$,
formalized by $\forall \rho.[\;\rho\comply\sigma \implies \exists f.\; \rho\complyP_f \sigma']$, is decidable and the orchestrator $f$ can be inferred. Moreover, the $f$ inferred is the same for any possible client.\\

The great flexibility of Padovani's notion of weak compliance is not for free, since the necessity of 
having an orchestrator for any  session of interactions between a client and a server could 
result sometimes in a heavy workload of the system.

In the present paper (taking into account first-order session-behaviours) we investigate a possible way of 
relaxing the standard notion of compliance such that interactions can still be carried on without any
orchestrator or mediating process. We do that by allowing the
possibility for some actions on the server's side to be {\em skipped}.

\noindent
Let us consider the following behaviour for the ballot service, where $\Dual{\tt InfoW}$
is a message providing informations about why a login has not been accepted.

\begin{tabbing}
${\sf Ball} \ByDef$ \= $\rec x.\;{\tt Login}.(\Dual{\tt Wrong}.\Dual{\tt InfoW}.x\;\oplus\; \Dual{\tt Overload}.x \;\oplus\; \Dual{\tt Ok}.\Dual{\tt Id}.($\= \kill

${\sf BallotServiceBehSkp} \ByDef $\\
   \> $\rec x.\;{\tt Login}.(\Dual{\tt Wrong}.\Dual{\tt InfoW}.x\;\oplus\; \Dual{\tt Overload}.x \;\oplus\; \Dual{\tt Ok}.\Dual{\tt Id}.(\;\;\;\;{\tt VoteA}.({\tt Va1} + {\tt Va2})$\\
             \>                                  \>$+\; {\tt VoteB}.({\tt Vb1} + {\tt Vb2})$ \\
             \>                                  \>$+ \; {\tt VoteC}.({\tt Vc1} + {\tt Vc2})\;)\;)$
\end{tabbing}
\Comment{
\begin{tabbing}
${\sf BallotServiceBehSkp} \ByDef$ \= $\rec x.\;{\tt Login}.($\=$\Dual{\tt Ok}.\Dual{\tt Id}.($\= \kill

${\sf BallotServiceBehSkp} \ByDef \rec x.\;{\tt Login}.(\Dual{\tt Wrong}.\Dual{\tt InfoW}.x$\\
                                        \>                  \> $\;\oplus\; $\\
                                        \>  \> $\Dual{\tt Overload}.x $ \\
                                         \> \> $\;\oplus\; $\\
                                         \>\> $\Dual{\tt Ok}.\Dual{\tt Id}.($\=${\tt VoteA}.({\tt Va1} + {\tt Va2})$\\
             \>                \>                  \>$+$\\
             \>                \>                  \>${\tt VoteB}.({\tt Vb1} + {\tt Vb2})$ \\
             \>                \>                  \>$+$\\
             \>                \>                  \>${\tt VoteC}.({\tt Vc1} + {\tt Vc2})) $
                                                     $)$
\end{tabbing}
}

Moreover, let the following be the behaviour of a possible voter,

\begin{tabbing}
${\sf VoterBeh2} = \rec  x.\;\Dual{\tt Login}.({\tt Wrong}.\;x + {\tt Overload}.\;x + {\tt Ok}.\Dual{\tt VoteB} ) $
\end{tabbing}

It is possible to check that ${\sf VoterBeh2} \not\complyP {\sf BallotServiceBehSkp}$, that is ${\sf VoterBeh2}$ is not compliant with ${\sf BallotServiceBehSkp}$ according to  
the Padovani's weak compliance, notwithstanding all the requests of the client can be
satisfied in their precise order and the only extra requirement needed is just the possibility for some particular action of the server to be {\em skipped}.
The problem lies in the definition of orchestrator: it does not allow unbounded buffering of
messages (or, equivalently, it does not allow skipping an unbounded number of messages coming of the server.)\\
So, in order to have ${\sf VoterBeh2}$ to be compliant with${\sf BallotServiceBehSkp}$ a different
definition of compliance is needed (which we call {\em skip-compliance}). Such a definition, besides, does not need to rely at
all on a notion of orchestrator. Even more, the notion of orchestrator can be avoided also during the interaction between a client and a server.

Actually, also Padovani's orchestrator formalism in \cite{Padovani10} allows  for a number of actions, either on the client's or the server's side, to be  {\em skipped} during a session of interactions. Such an effect can be achieved by simply keeping some messages from the server inside the  buffer, without without ever passing them to the client.
The number of skipped actions (hided ones, actually) is however necessarily bounded.
In fact the buffering capabilities of orchestrators in the framework
of  \cite{Padovani10} are bounded (orchestrators need to be of {\em finite} rank.)
The formalism of the present paper allows for an unbounded number of skipping actions,
We shall present later on an example showing the usefulness of such a property.\\
The $\skipAct$-compliance relation we investigate in the present paper can then 
be looked at as a minimal weakening of the standard notion of compliance not requiring the introduction of {\em orchestrators} between compliant clients and servers.\\

The formal definition of $\skipAct$-compliance we provide in the present paper, beside relying, as usual, on two LTSs
(one on session-behaviours and one on client/server pairs), needs to make sure that
the traces out of client/server pairs do not contain infinitely many
$\skipAct$-actions. Such a condition enable only sound recursive session-behaviours to be $\skipAct$-compliant.\\
Our compliance relation can be coinductively characterized and proven to be decidable by means of a formal system implicitly corresponding to a decision algorithm. The correctness proof of the decision algorithm can be developed along the lines of the similar proof
of decidability of the sub-behaviour relations in \cite{BdL13}, which is, in turn, an extension
of the decision algorithm for subtyping in \cite{GH05}.\\ 
A $\skipAct$-subbehaviour relation, $\preceqG$, can then be defined. Unlike for the standard sub-behaviour relation, the property of the dual of a behaviour $\sigma$ being the minimum among  
those {\em preceding} $\sigma$ turns out to be far more involved to prove, due the presence of skipped actions.
Such a property will directly imply decidability for the $\preceqG$ relation, so enabling us to
use is in type systems with session types, being our session-behaviours in one-to-one correspondence with session types.

When briefly discussing about future work in Section \ref{sect:futureWork}, we shall
hint at the possibilty of using Padovani weak-compliance to model interactions in Erlang-like
languages, and our $\skipAct$-compliance to model that part of interactions related to 
{\sf receive} operations in the style of Erlang. \\

\fi 

\section{Session Behaviours and the $\skipAct$-compliance relation}
\label{sect:skipCompliance}
Contracts \cite{LP07,CGP10} are a subset of CCS terms, defined by the grammar :
\[\sigma ::= \stopA \mid \alpha.\sigma \mid \sigma + \sigma \mid \sigma \oplus \sigma \mid x \mid \rec x.\sigma \]
where $\alpha$ ranges over a set of actions and co-actions, $\stopA$ is the same as the CCS term $0$, namely the completed protocol, $+$ and $\oplus$ are external and internal choices respectively.
{\em Session behaviours} as defined below are a further restriction of this set. They are designed to be in one-to-one correspondence to session types \cite{honda.vasconcelos.kubo:language-primitives} without delegation (in \cite{BdL10} and \cite{BdL13} session behaviours were extended by send/receive actions of session behaviours to model delegation). The restriction is achieved by constraining internal and external choices in a way that limits the non-determinism to (internal) output selection.


\begin{definition}[Session Behaviours]\label{def:session-behaviours}\hfill
\begin{enumerate}[i)]
\item
Let $\Names$ be some countable set of symbols and~ $\CoNames = \Set{\Dual{a} \mid a \in \Names}$, with
$\Names\cap\CoNames = \emptyset$. \\
The set $\Rbehav$ of {\em raw behaviour expressions} is defined by the following grammar:
\[\begin{array}{lcl@{\hspace{4mm}}l}
\sigma,\tau & ::= & ~~ \stopA & \mbox{inaction} \\
       &     & \mid ~ a_1.\sigma_1 + \cdots + a_n.\sigma_n  & \mbox{external choice} \\
       &     & \mid ~\Dual{a}_1.\sigma_1 \oplus \cdots \oplus \Dual{a}_n.\sigma_n & \mbox{internal choice}\\
       &     & \mid  ~x  & \mbox{variable}\\
       &     & \mid ~\rec x. \sigma &  \mbox{recursion}
\end{array}
\]
where
\begin{itemize}
\item [-]
 $n \geq 1$ and $a_i \in\Names$ (hence $\Dual{a}_i\in\CoNames$) for all $\;1 \leq i\leq n$;
\item[-]
$x$ is a session behaviour variable out of a denumerable set and it is bound by the  $\rec$ operator.
\end{itemize}

As usual, $\sigma$  is  said to be {\em closed} whenever $\fv{\sigma}= \emptyset$,
where $\fv{\sigma}$
denotes the set of free variables in $\sigma.$

\vspace{2mm}
\item
The set $\Sbehav$ of {\bf session behaviours} is the subset of closed raw behaviour expressions such that in $a_1.\sigma_1 + \cdots + a_n.\sigma_n$ and $\Dual{a}_1.\sigma_1 \oplus \cdots \oplus \Dual{a}_n.\sigma_n$,
 the $a_i$  and  the $\Dual{a}_i$ are, respectively, pairwise distinct; moreover in $\rec x.\sigma$ the expression $\sigma$ is not a variable.\\
\end{enumerate}
\end{definition}
We abbreviate  $a_1.\sigma_1 + \cdots + a_n.\sigma_n$ by $\sum_{i=1}^n a_i.\sigma_i$, and
$\Dual{a}_1.\sigma_1 \oplus \cdots \oplus \Dual{a}_n.\sigma_n$ by $ \bigoplus_{i=1}^n \Dual{a}_i.\sigma_i$. 
We also use the notations $\sum_{i\in I}a_i.\sigma_i$ and $\bigoplus_{i\in I} \Dual{a}_i.\sigma_i$, for finite and not empty $I$.
The trailing $\stopA$ is normally omitted: we write e.g. $a+b$ for $a.\stopA+ b.\stopA$. 

Note that recursion in $\Sbehav$ is guarded and hence contractive in the usual sense \cite{BD91}.
Session behaviours will be considered modulo commutativity of internal and external choices.\\

A syntactical notion of {\em duality\/} on $\Sbehav$ is easily obtained by 
interchanging $a$ with $\Dual{a}$, and $+$ with $\oplus$. Its formal definition can obtained by restricting to $\Sbehav$ a straightforward definition by induction on the structure
of the raw expressions in $\Rbehav$ (i.e. also for open expressions\footnote{To avoid too cumbersome definitions, any time an inductive definition
on elements of $\Sbehav$ will be provided, it will be tacitly assumed to be actually the restriction
to $\Sbehav$ of the corresponding inductive definition on $\Rbehav$.}). 
The dual of a session-behaviour $\sigma$ will be denoted, as usual, by $\Dual{\sigma}$.
As expected, $\Dual{\Dual{\sigma}} = \sigma$ for all $\sigma$. 

\medskip
The operational semantics of session behaviours is given in terms of a labeled transition system (LTS) 
$\sigma \lts{\alpha} \sigma'$ where $\sigma,\sigma'\in\Sbehav$ and $\alpha$ belongs to an appropriate set of actions $\Act$.

\begin{definition}[Behaviour LTS] \label{def:SB-red}~\hfill\\
Let $\skipAct\not\in\Names$ and define the set of actions $\Act=  \Names\cup\CoNames$ and $\IntChoice,\unfoldAction\not \in \Act$; then
define the LTS $(\Sbehav,   \Act \cup \Set{\IntChoice,\unfoldAction}, \lts{})$ by the rules:
\[\begin{array}{c@{\hspace{12mm}}c@{\hspace{6mm}}c@{\hspace{6mm}}c}
a_1.\sigma_1 + \cdots + a_n.\sigma_n \lts{a_k} \sigma_k &
\Dual{a}.\sigma
\lts{\Dual{a}} \sigma \vspace{3mm}\\
\Dual{a}_1.\sigma_1 \oplus \cdots \oplus \Dual{a}_n.\sigma_n
\lts{\IntChoice} \Dual{a}_k.\sigma_k
  & \rec x. \sigma \lts{\unfoldAction}
\sigma\Subst{\rec x.\sigma}{x}
\end{array}\]
where $1 \leq k \leq n$ and $\sigma \lts{\alpha}\gamma$ abbreviates $(\sigma,\alpha,\gamma) \in\; \lts{}$. 
\end{definition}

We abbreviate $\lts{} \, = \, \lts{\IntChoice} \cup
\lts{\unfoldAction}$. Note that neither $\IntChoice$ nor
$\unfoldAction$ are actions, so that they are unobservable and used just for technical reasons; indeed we adopt the
standard $\lts{}$ (from CCS without $\tau$) in the subsequent
definition of the parallel operator for testing.
As usual, we write $\Lts{} = \lts{}^*$ and $\Lts{\alpha} = \lts{}^*
\lts{\alpha} \lts{}^*$ for $\alpha\in\Act$.

We observe that if $\sigma\Lts{\alpha}\sigma'$ or $\sigma\Lts{}\sigma'$ for $\sigma\in\Sbehav$, then $\sigma'\in\Sbehav$.

\begin{lemma}\label{lem:convergence}
For any $\sigma \in \Sbehav$ there exists a unique and finite set $R = \Set{\sigma' \in \Sbehav \mid \sigma \Lts{} \sigma'\not\Lts{}}$, which is either of shape $\Set{\stopA}$ or $\Set{a_1.\sigma_1 + \cdots + a_n.\sigma_n}$ or $\Set{\Dual{a}_i .\sigma_i \mid i\in I}$. Moreover $R$ is computable in $\sigma$.
\end{lemma}

\begin{proof} By induction of the structure of $\sigma$.
Since recursion is guarded and internal choices are finitary, no infinite $\lts{}$ reductions are possible out of $\sigma$; on
the other hand if $\sigma\in\Sbehav$ then it is closed, so the case $\sigma \Lts{} x$ for some variable $x$ is impossible.
\end{proof}

In the sequel we write $\Converge{\sigma}{\stopA}$ and   $\Converge{\sigma} {\sum_{i\in I} a_i.{\rho}_i}$ if
the $R$ in the above lemma is, respectively, of the first two shapes, and write $\Converge{\sigma} {\bigoplus_{i\in I} \Dual{a}_i.{\sigma}_i}$ if $R = \Set{\Dual{a}_i .\sigma_i \mid i\in I}$.

\medskip
We shall denote finite or infinite sequences of elements of $\Act$, i.e. elements of  $\Act^*\cup\Act^\infty$, by bold characters $\vec{\alpha}, \vec{\beta}, \ldots$. Bold italic (overlined) characters $\vec{a}, \vec{b}, \vec{c},\ldots$ ($\Dual{\vec{a}}, \Dual{\vec{b}}, \Dual{\vec{c}},\ldots$) shall denote sequences of elements of $\Names$ (resp. $\CoNames$). We shall represent the fact that a sequence $\vec{\alpha}$ is infinite by writing $\vec{\alpha}^\infty$. 
The length of a sequence $\vec{\alpha}$ will be denoted by $|\vec{\alpha}|$, and it is either finite or $\infty$.

We write  $\sigma\Lts{\vec{\alpha}} \sigma'$ if $\vec{\alpha} = \alpha_1\cdots \alpha_n$ and
$\sigma\Lts{\alpha_1}\cdots\Lts{\alpha_n}\sigma'$.
Also we write $\sigma\lts{}$ and $\sigma\lts{\alpha}$ if there exists $\sigma'$ s.t.
$\sigma\lts{}\sigma'$ and $\sigma\lts{\alpha}\sigma'$ respectively, and $\sigma\not\lts{}$
when $\neg (\sigma\lts{})$.
Given $\vec{\alpha}=\alpha_1\ldots\alpha_n$ the notation $\beta\in\vec{\alpha}$ will stand for $\beta\in\{\alpha_1,\ldots,\alpha_n\}$.\\

We define the set of traces of a session behaviour as follows. 
\begin{definition}[Traces]
\label{def:traces}
The mapping $\trace :  \Sbehav \rightarrow ({\cal P}(\Act^*)\cup{\cal P}(\Act^\infty))$ is defined by
$$\begin{array}{l@{\hspace{12mm}}l}
\trace(\sum_{i\in I} a_i.{\sigma}_i)=\bigcup_{i\in I} \{a_i\,\vec{\alpha} \mid \vec{\alpha}\in \trace(\sigma_i)\} &
\trace(\stopA) = \{\varepsilon\} \\[2mm]
\trace(\bigoplus_{i\in I} \Dual{a}_i.{\sigma}_i)=\bigcup_{i\in I} \{\Dual{a}_i\,\vec{\alpha} \mid \vec{\alpha}\in \trace(\sigma_i)\} &
\trace(\rec x.\sigma) = \trace(\sigma[\rec x.\sigma/x])
\end{array}
$$ 
\end{definition} 
\noindent
A session-behaviours $\sigma$ is said to be {\em finite} whenever
$\trace(\sigma)\in {\cal P}(\Act^*)$.

\subsection{The $\skipAct$-compliance relation}
\label{subsect:skipCompliance}

As for contract compliance, we use an LTS of client/server pairs $\rho\|\sigma$ to define the notion of $\skipAct$-compliance on session-behaviours.
The actions of the LTS are the silent action $\tau$, representing a full handshake between
synchronizing actions on the client and server sides, together with a  ``skipping" action $\skipAct$,
representing the fact that an action on the server's side has been discarded.

As mentioned in the introduction, we allow only output actions on the server side to be discarded. However we disallow the skip of an output action that synchronizes with some input action by the client. Let us write:
\[\noLts{\rho}{\alpha} \Iff \neg \, \exists\, \rho'.~\rho\Lts{\alpha}\rho'.\]
Observe that the statement $\noLts{\rho}{\alpha}$ is decidable because it is the negation
of $\Converge{\sigma}{\sum_{i\in i}a_i.\sigma_i}$ or $\Converge{\sigma}{\bigoplus_{i\in i}\Dual{a}_i.\sigma_i}$, with $\alpha\in \Set{a_i,\Dual{a}_i \mid i\in I}$,
which are decidable by Lemma \ref{lem:convergence}.

\if \false
One possibility is to allow any action on the server side that the client is not able to match:
\[
\prooftree 
\noLts{\rho}{\Dual{\alpha}} \quad \sigma\lts{\alpha}\sigma'  
\justifies
\rho\|\sigma\lts{\skipAct}\rho\|\sigma'\endprooftree
\]
By this rule we have that $a+b$ is $\skipAct$-compliant with $d.\Dual{a} + c.\Dual{b}$, for example.

The definition of the LTS rule concerning the $\skipAct$ action need to be properly discussed.\\
Notice that the naive following rule would not result in a meaningful system.
\vspace{-3mm}
$$
\prooftree \rho\lts{\alpha}\rho' \quad
\sigma\lts{\beta}\sigma'   \quad \beta\not=\Dual{\alpha}\justifies
\rho\|\sigma\lts{\skipAct}\rho\|\sigma'\endprooftree
$$

\vspace{-3mm}
In fact that would imply the following behaviours to be compliant:
$a+b \complyG d.\Dual{a} + c.\Dual{b}$\\
The interaction between $a+b$ and $d.\Dual{a} + c.\Dual{b}$, however, can be carried on
only in case an orchestrator be present. Moreover, such an orchestrator should not only coordinate the messages exchange, 
but should also play a heavy role in the evolution of $d.\Dual{a} + c.\Dual{b}$, since it should 
decide which one,  among the possible branches, can evolve.
We wish to avoid that.
(Notice that, as mentioned before, such a possibility is also  avoidedd  in \cite{Padovani10}, where a sound orchestrator cannot send messages other than those that it has received by the client or the server.) \\
As an extra counterexample to the feasibility of  the use of the above rule, it would result in a sub-behavior relation among servers that would allow ~
$\Dual{a}\oplus \Dual{b}\oplus \Dual{c} \preceq_s d(\Dual{a}\oplus \Dual{b})+c$\\
This would be hardly meaningful if we used $\preceq_s$ as the basis of a subtyping relation
on session types.\\


Also the following LTS rule would be problematic:
\vspace{-3mm}
$$
\prooftree \rho\lts{\alpha}_0\rho' \quad
\sigma\lts{\Dual{\beta}}\sigma'   \quad \beta\not=\alpha\justifies
\rho\|\sigma\lts{\skipAct}\rho\|\sigma'\endprooftree
$$

\vspace{-2mm}
\noindent
A motivation similar to the one for the first rule mentioned would apply.  In fact the present rule would allow ~
$\Dual{a}\oplus \Dual{b} \complyG c.a+d.b$~
where, again, the presence of an orchestrator should be needed: it would send a message $\Dual{c}$ or a message $\Dual{d}$
according to what choice is made by $\Dual{a}\oplus \Dual{b}$. Even if such a
behaviour for an orchestrator is allowed for the orchestrators as described in \cite{Padovani10}, the message asked for by the server could be relevant for the evolution of the server itself. So, beside the fact that we wish to relax the notion of compliance as far as the presence of a coordinator can be avoided, we decide not to allow the above possibility. 

So a reasonable rule to be used is the following one
\vspace{-3mm}
$$
\prooftree \rho\lts{\alpha}\rho' \quad
\sigma\lts{\Dual{b}}\sigma'   \quad \alpha\not=b\justifies
\rho\|\sigma\lts{\skipAct}\rho\|\sigma'\endprooftree
$$

\vspace{-2mm}
\noindent
Notice that the above  rule conforms not only to \cite{Padovani10} but also to Mostrous treatment of asyncronous subtyping \cite{DM09}.

It is worth pointing out that, in case we took out the side condition $\alpha\not=b$, the presence of an orchestrator would turn out to be necessary, since we could be able to skip also synchronizing actions,
so adding extra nondeterminism.
In such a case, in fact, we could have $a \complyG \Dual{b}.\Dual{a}.\Dual{b}.\Dual{a}$ because
$$a \| \Dual{b}.\Dual{a}.\Dual{b}.\Dual{a} \lts{\skipAct} a \| \Dual{a}.\Dual{b}.\Dual{a} \lts{\tau} \stopA \| \Dual{b}.\Dual{a}$$ But also
$$a \| \Dual{b}.\Dual{a}.\Dual{b}.\Dual{a} \lts{\skipAct} a \| \Dual{a}.\Dual{b}.\Dual{a} \lts{\skipAct} a \| \Dual{b}.\Dual{a} \lts{\skipAct} a \| \Dual{a}\lts{\tau} \stopA \|  \stopA$$
 So, the actual interaction between $a$ and 
$\Dual{b}.\Dual{a}.\Dual{b}.\Dual{a}$ should be carried on necessarily by an {\em orchestrator}.\\

We need now to decide whether the number of {\em consecutive} $\skipAct$ actions need to be "bounded"
or can be "unbounded". Let us consider the following example:~
$ b\complyG \rec x.(\Dual{a}.x\oplus \Dual{b})$\\
In such an example the server would be compliant 
with the client only in case the former were guaranteed to be {\em fair}.
Such, however, is not a guarantee we can infer from what we know by the behaviour
exposed by the server.
We decide hence to forbid the server $\rec x.(\Dual{a}.x\oplus \Dual{b})$ to be compliant with the client $b$.

We wish then to prevent the possibility of an unbounded number (and hence an infinite one, 
in presence of unfair servers) of {\em consecutive} $\skipAct$ actions. This however does not means that the {\em overall} number of $\skipAct$ action need to be bounded.
In fact the compliance in the following example does not depend on any fairness property on the
server's side, but nonetheless an infinite number of skipping actions is required:
\vspace{-3mm}
\begin{equation}
\label{eq:example1}
\rec x. \Dual{b}.x \complyG \rec x.\Dual{a}.\Dual{a}.b.x
\end{equation}

\vspace{-2mm}
\noindent
Our definition of compliance will then prevent the possibility of an unbounded number
of {\em consecutive} skipping actions, but will allow an unbounded {\em overall} number of them.\\
We are not saying that allowing unbounded consecutive skipping actions be completely unreasonable,
but simply that it would force us to rely on a fairness condition on the server's side. Besides, it would result also in 
in a computationally heavy coinductive compliance relation, and decidability conditions.

\fi

\medskip
The next definitions formally introduce the LTS for client/server pairs and the relation of  $\skipAct$-{\em compliance} 
for session behaviours, that we dub $\complyG$. 

\begin{definition}[LTS for Client-Server pairs]\label{def:commCompl}
~\\
Let $\csAct = \{\tau,\skipAct\}$ be the set of the synchronization actions and
$\rho\|\sigma$ denote the parallel composition of session behaviors in $\Sbehav$, then define:
\[\begin{array}{c@{\hspace{8mm}}c}
\prooftree \rho\lts{}\rho' \justifies
\rho\|\sigma\lts{}\rho'\|\sigma\endprooftree &
\prooftree \sigma\lts{}\sigma'\justifies
\rho\|\sigma\lts{}\rho\|\sigma'\endprooftree
\\[6mm]
\prooftree
	\noLts{\rho}{a}  \quad
	\sigma\lts{\Dual{a}}\sigma'   
\justifies
	\rho\|\sigma\lts{\skipAct}\rho\|\sigma'
\endprooftree
&
\prooftree \rho\lts{\alpha}\rho' \qquad
\sigma\lts{\Dual{\alpha}}\sigma'   \justifies
\rho\|\sigma\lts{\tau}\rho'\|\sigma'\endprooftree
\end{array} 
\]
where $a\in\Names$ (and hence $\Dual{a}\in\CoNames$), $\alpha\in\Act$ and $\Dual{\alpha}$  is its dual, such that $\Dual{\Dual{\alpha}} = \alpha$.
\end{definition}
The ratio of introducing the ability of clients to skip some actions on the server side is to allow more clients to synchronize with servers that essentially provide the required service but for some supplementary (and possibly redundant) information. \\

We abbreviate $\Lts{}~=~\lts{}^*$
and $\Lts{\xi}~=~\Lts{}\circ\lts{\xi}\circ\Lts{}$, where $\xi\in\csAct$.\\
Moreover, by $\Lts{\zeta^*\xi}$ we denote $\Lts{\zeta}$$^*\circ\Lts{\xi}$, where $\zeta,\xi\in\csAct$.

\begin{remark}\label{rm:skip-motivation}
{\em 
We observe that it would be unreasonable to allow clients to deny replies to server input actions, as this would result into a complete loss of control (think of the {\tt Login} action in the ballot service examples). On the other hand we balance the possibility of skipping server outputs by two principles. The first one is that the client is not
allowed to defer the synchronization with an output action of the server which it is ready to accept, avoiding the indeterminacy of synchronizations like 
$$a \| \Dual{b}.\Dual{a}.\Dual{b}.\Dual{a} \lts{\skipAct} a \| \Dual{a}.\Dual{b}.\Dual{a} \lts{\tau} \stopA \| \Dual{b}.\Dual{a}$$ and
$$a \| \Dual{b}.\Dual{a}.\Dual{b}.\Dual{a} \lts{\skipAct} a \| \Dual{a}.\Dual{b}.\Dual{a} \lts{\skipAct} a \| \Dual{b}.\Dual{a} \lts{\skipAct} a \| \Dual{a}\lts{\tau} \stopA \|  \stopA$$
of which only the first one is legal. The second principle is that a client has not to be compliant with a server that will never provide the required output. This happens in an infinite interaction which is {\em definitely} made of $\skipAct$-synchronization actions, as in the cases of $\rec x.b.x \| \rec x.\Dual{a}.x$ and of the subtler $b \| \rec x.(\Dual{a}.x \oplus \Dual{b})$.\\
However, it is reasonable to allow the overall number of skippings to be infinite. A simple example of that is when all the infinite  $b$'s of the client $\rec x.b.x$ manage  to syncronize with a $\Dual{b}$ of the server $\rec x.\Dual{a}.\Dual{a}.\Dual{b}.x$, each time skipping the  $\Dual{a}$ preceding the $\Dual{b}$ and the $\Dual{a}$ following it.
}
\end{remark}

So, as previously discussed, the notion of compliance we wish to formalize is an extension
of the usual notion of compliance such that any finite or infinite number of output actions from the server can be discarded. We wish however to rule out the possibility of a client indefinitely discarding output actions coming from the server. 
So, in order to do that, we formalize below the $\skipAct$-compliance relation in terms of 
{\em synchronization traces}. A synchronization trace describes a possible client/server interaction as a sequence of successful handshakes ($\tau$) or skipping actions ($\skipAct$).
Such traces can be either finite or infinite. A client will then be compliant with a server when all the 
client/server finite synchronization traces ends with $\checkmark$ (which can occur only in case
the client completes, i.e. gets to $\stopA$) and all the infinite synchronization traces are not formed of just $\skipAct$ elements from a certain point on, i.e. are not {\em definitely}-$\skipAct$.

\begin{definition}[Synchronization traces]\label{def:synch-traces}
~\\
The mapping 
$\cstrace :  \Sbehav\times\Sbehav \rightarrow ((\csAct\cup\{\checkmark\})^*\cup\csAct^\infty)$ is defined by
\begin{itemize}
\item[~]
$\cstrace(\rho\|\sigma)=
          \left\{\begin{array}{l@{~~~~~}l}
            \{\checkmark\} & \mbox{if~~}  \rho=\stopA \\[2mm]
            \{\xi\,\vec{\chi} \mid \rho\|\sigma \Lts{\xi}\rho'\|\sigma' ~ \And \vec{\chi}\in\cstrace(\rho'\|\sigma')\} & \mbox{if~~}  \exists\zeta\in\csAct.~\rho\|\sigma \Lts{\zeta}\\[2mm]
           \{\varepsilon\} & \mbox{otherwise} 
            \end{array}
          \right.$
\end{itemize}
Let $\vec{\xi}\in\csAct^\infty$ with $\vec{\xi}=\xi_1\xi_2\ldots$. We say $\vec{\xi}$ to be {\em definitely}-$\skipAct$ whenever $\exists k.~ \forall h>k.~ \xi_h=\skipAct$.
\end{definition}

\Comment{ 
The following properties should be relevant.

\begin{proposition}\hfill
\label{prop:PropertiesSkip}
\begin{enumerate}[i)]
\item
Let $\alpha\in\Act$ and let $\vec{b}\in \Act^*\cup\Act^\infty$ such that    $\sigma\Lts{\Dual{\vec{b}}}$. Then
$$
\alpha\in \vec{b} ~~~ \Iff ~~~ \exists \vec{c}.~ (\sigma\Lts{\Dual{\vec{c}}} \;\And \alpha\in \vec{c} \And |\vec{c}|\leq \mbox{number related to the Parallel Unfolding of recs in $\sigma$})$$
\item
$$a\complyG \sigma ~~~ \implies ~~~ \not\exists\vec{c}^\infty.~ (\sigma\Lts{\Dual{\vec{c}}^\infty} \And a\not\in \vec{c}^\infty)$$
\end{enumerate}
\end{proposition}
}
Then the notion of $\skipAct$-compliance can be formalised in terms of synchronization traces
as follows.
\begin{definition}[$\skipAct$-compliance]\label{def:skipcompl}
The {\em client} $\rho$
 is {\em skip-compliant} with the {\em server} $\sigma$, written $\rho \complyG \sigma$, whenever,
for any $\vec{\xi}\in\cstrace(\rho \| \sigma)$ either $\vec{\xi}=\vec{\xi}'\checkmark$ or 
$\vec{\xi}$ is infinite and not definitely-$\skipAct$.
\end{definition}
In the remaining part of the paper we just say ``compliant" instead of ``$\skipAct$-compliant" when any ambiguity cannot arise.

By the previous definition we have that, as stated in the Introduction, ${\sf Voter} \complyG{\sf BallotServiceBehSkp}$. In the following example we provide, instead, two behaviours that are not compliant.
\Comment{
Let us provide now a couple of examples of compliant and not compliant client/server pairs.
\begin{example}
{\em
By the previous definition we have ${\sf Voter} \complyG{\sf BallotServiceBehSkp}$, since
any element of the (infinite) set $\cstrace({\sf Voter} \| {\sf BallotServiceBehSkp})$ is 
\begin{itemize}
\item[-] either a finite sequence of the form
$\ldots\tau\tau\skipAct\tau\checkmark$;\\ 
This happens when the voter logs-in a finite number of times. Each time, but the last, the voter provides a wrong login (and skipping the information provided by the server
about why the login was refused) or finds
the server overloaded. When the voter logs-in for the last time, after receiving an ${\tt ok}$ message, skips the identifier of the transaction and completes after voting the candidate $A$.
\item[-] or an infinite sequence $\tau\ldots\tau\ldots$ not containing any  $\skipAct$;\\
This happens when the server keeps indefinitely on being overloaded. In such a case such an infinite sequence cannot be, trivially,
a {\em definitely}-$\skipAct$ one, since it contains no $\skipAct$ at all.
\item[-] or an infinite sequence $\tau\ldots\tau\skipAct\tau\ldots$ containing a finite or infinite number of $\skipAct$.\\
This happens when the client tries an infinite number of accesses and at any time either the server is overloaded or the login wrong.  In such an infinite sequence any $\skipAct$ cannot be the one due to the skipping of the transaction identifier. Moreover,  any $\skipAct$  is necessarily preceded  by two $\tau$'s, since it corresponds to the skipping of the informations about why the logging procedure failed and this skipping is always preceded by the login attempt and the reception of the ${\tt Wrong}$ message. This last fact then implies  that such an infinite sequence cannot be a {\em definitely}-$\skipAct$ one.
\end{itemize}
}
\end{example}
}

\begin{example}
{\em
Let us consider the following malicious server that, after receiving a login, sends a ${\tt Wrong}$ message and then, indefinitely, the message ${\tt InfoW}$, that is 
$${\sf BallotServiceMalicious} \ByDef 
    {\tt Login}.(\Dual{\tt Wrong}.\rec x.\Dual{\tt InfoW}.x)$$
It is easy to check that ${\sf Voter} \not\complyG{\sf BallotServiceMalicious}$.\\
 In fact we have that $\cstrace({\sf Voter} \| {\sf BallotServiceMalicious})=\{\tau\,\tau\,\skipAct\,\skipAct\,\skipAct\ldots\}$, that is the only element of the set of synchronization traces is a sequence that, after the two $\tau$ actions due to the login message and the message that the login procedure went wrong, is made of an infinite number of consecutive $\skipAct$'s, since the server would keep on skipping all the ${\tt InfoW}$ messages from the server. Such a sequence is an obviously {\em definitely}-$\skipAct$ one.
}
\end{example}

\begin{remark}\label{rem:strongComplInclWeak}
{\em
It is clear that $\comply \; \subseteq \; \complyG$. This inclusion is strict: in fact $b \complyG \Dual{a}.\Dual{b}$ with $a\neq b$, but
$b \not\comply  \Dual{a}.\Dual{b}$.
}\end{remark}


\if false
Going back to our Ballot Service example, we have now that ~
${\sf VoterBeh2} \complyG{\sf BallotServiceBehSkp}$
Notice that this is a counterexample showing the following

\begin{proposition} 
\label{prop:counterPadovani}
Let $\complyP$ be Padovani's weak $k$-compliance restricted to session behaviours.
Then, for any $k$, we have:
\vspace{-3mm} 
$$\complyG ~\not\subseteq ~\complyP_k$$
\end{proposition}
\noindent
The relation $\complyP_k$  is defined in \cite{Padovani10}:  in terms of : $\rho \complyP_k\sigma$ whenever there exists an orchestrator of finite rank $k$ such that
$\rho \complyP_f\sigma$, where the rank of an orchestrator $f$ is  the maximum capacity of the buffer associated to $f$.\\
Now, the actual interaction between ${\sf VoterBeh2}$ and ${\sf BallotServiceBehSkp}$ should be carried on in Padovani's setting through the use of the following orchestrator:
\begin{tabbing}
${\mathit BallotServiceOrch} \;\ByDef \;$ \= $\rec x.\;\orchAct{{\tt Login}}{{\tt Login}}.$\=$\Dual{\tt Ok}.\Dual{\tt Id}$\= \kill

${\mathit BallotServiceOrch} \;\ByDef \; \rec x.\;\orchAct{{\tt Login}}{\Dual{\tt Login}}.(\;\;\;\;\orchAct{\Dual{\tt Wrong} }{{\tt Wrong}}.\orchAct{\varepsilon }{{\tt InfoW}}.x$\\
                                        \>                  \> $\;\vee\; \orchAct{\Dual{\tt Overload} }{{\tt Overload}}.~x $ \\
                                         \> \> $\;\vee\; \orchAct{\Dual{\tt Ok} }{{\tt Ok}}.\orchAct{\varepsilon }{{\tt Id}}.$\=$\orchAct{{\tt VoteB}}{\Dual{\tt VoteB}} \;)$
\end{tabbing}
\Comment{
\begin{tabbing}
${\mathit BallotServiceOrch} =$ \= $\rec x.\;\orchAct{{\tt Login}}{{\tt Login}}.($\=$\Dual{\tt Ok}.\Dual{\tt Id}.($\= \kill

${\mathit BallotServiceOrch} = \rec x.\;\orchAct{{\tt Login}}{\Dual{\tt Login}}.(\orchAct{\Dual{\tt Wrong} }{{\tt Wrong}}.\orchAct{\varepsilon }{{\tt InfoW}}.x$\\
                                        \>                  \> $\;\vee\; $\\
                                        \>  \> $\orchAct{\Dual{\tt Overload} }{{\tt Overload}}.~x $ \\
                                         \> \> $\;\vee\; $\\
                                         \>\> $\orchAct{\Dual{\tt Ok} }{{\tt Ok}}.\orchAct{\varepsilon }{{\tt Id}}.($\=$\orchAct{{\tt VoteB}}{\Dual{\tt VoteB}}$
                                                     $)$
\end{tabbing}
}
which is not of finite rank.

The behaviours of (\ref{eq:example1}) provide another counterexample, since
we have that $\rec x. \Dual{b}.x \not\complyP_k \rec x.\Dual{a}.\Dual{a}.b.x$ for any $k$
since no orchestrator of finite rank can mediate between $\rec x. \Dual{b}.x$ and $\rec x.\Dual{a}.\Dual{a}.b.x$.
The interaction can be carried on only by the orchestrator $f = \rec x. \orchAct{\varepsilon}{a}.\orchAct{\varepsilon}{a}.\orchAct{\Dual{b}}{b}.x$.
Notice that Proposition \ref{prop:counterPadovani} above heavily relies on the possibility, in
our formalism, of
having an infinite number of $\skipAct$-actions during a client/server interaction. In fact, for finite behaviours, the following holds. 

\begin{proposition} For any pair of {\em finite} session behaviours $\rho,\sigma$, there exists a $k\geq0$
such that 
\vspace{-3mm}
$$\rho\complyG\sigma ~\implies ~\rho\complyP_k\sigma$$
\end{proposition}

\fi

\section{Coinductive characterization and decidability.}
\label{sect:coChar}

To prove that the $\complyG$ relation is decidable we work out a coinductive characterization.
In doing that we use the relation $\sync{}{}$ between actions and traces. $\sync{\alpha}{\sigma}$ holds whenever all traces
of the server $\sigma$ contain the action $\Dual{\alpha}$ possibly prefixed by a finite sequence of {\em skippable} output actions.


\begin{definition}[Coinductive Skip-Relations]\label{def:coskipcompl}
~

\begin{enumerate}[i)]
\item The relation $\sync{}{}\subseteq\Act\times\Sbehav$ is defined by
\begin{itemize}
\item[~]\vspace{1mm}
$\sync{a}{\sigma} ~\ByDef~ \forall \vec{\alpha}\in\trace(\sigma) \, \exists \,\Dual{\vec{b}}, \vec{\alpha}'.~
\vec{\alpha}= \Dual{\vec{b}}\,\Dual{a}\,\vec{\alpha}' \And \Dual{a}\not\in\Dual{\vec{b}}$;
\item[~]\vspace{2mm}
$\sync{\Dual{a}}{\sigma}~\ByDef~ \forall \vec{\alpha}\in\trace(\sigma) \, \exists \,\Dual{\vec{b}}, \vec{\alpha}'.~
\vec{\alpha}= \Dual{\vec{b}}\,a\,\vec{\alpha}'$,
\end{itemize}
where $\Dual{\vec{b}}$ is possibly empty.
\item \vspace{2mm}
\label{def:coskipcompl-ii}
The operator 
$\FunH: {\cal P}(\Sbehav\times \Sbehav)\rightarrow  {\cal P}(\Sbehav\times \Sbehav)$
is defined as follows: \\
 for any relation $\Rel\subseteq \Sbehav\times \Sbehav$, we have
$(\rho,\sigma) \in {\cal H}(\Rel)$ if and only if  either $\Converge{\rho}{\stopA}$ or the following statements hold:
\vspace {3mm}
\begin{enumerate}
\item $\Converge{\rho} {\sum_{i\in I} a_i.{\rho}_i} ~~\Rightarrow~~
          \left\{\begin{array}{l}
            \exists k\in I.~\sync{a_k}{\sigma} ~\And \\[2mm]
            \forall i\in I.\forall \sigma'.~[(a_i.\rho_i \|\sigma \Lts{\skipAct^*\tau} \rho_i\|\sigma') ~\implies~ \rho_i\Rel \sigma'] ;
            \end{array}
          \right.$
\item\vspace{3mm}
$\Converge{\rho} {\bigoplus_{i\in I} \Dual{a}_i.{\rho}_i} ~~\Rightarrow~~
          \left\{\begin{array}{l}
           \forall i\in I.~\sync{a_i}{\sigma}  \And\\[2mm]
            \forall i\in I.\forall \sigma'.~[(\Dual{a}_i.\rho_i \|\sigma \Lts{\skipAct^*\tau} \rho_i\|\sigma') ~\implies~ \rho_i\Rel \sigma'].
            \end{array}
          \right.$
\end{enumerate}
\vspace{2mm}
\item\vspace{1mm}
A relation $\Rel\subseteq \Sbehav\times\Sbehav$ is a
{\em coinductive Skip-relation\/} if and only if $\Rel\subseteq \FunH(\Rel)$.
\end{enumerate}
\end{definition}

\medskip Since  $\Rel$ occurs covariantly in the clauses defining $\FunH(\Rel)$, the operator
$\FunH$ is monotonic with respect to subset inclusion. Then the following fact  immediately follows by Tarsky theorem \cite{Tar55} (see also \cite{Sangio09} for a discussion about the use of this result): 

\begin{fact}\label{fact:gfpH}
Let $\FunH^0 \ByDef \Sbehav\times\Sbehav$ and $\FunH^{k+1} \ByDef \FunH(\FunH^k)$; then 
\[
\begin{array}{c}\nu(\FunH) ~ 
 = ~ \bigcup\Set{\Rel \subseteq \Sbehav\times\Sbehav\mid \Rel\subseteq \FunH(\Rel)} ~ 
 = ~ \bigcap_{k\in\Nat}\FunH^k
\end{array}
\]
is the greatest fixed point of $\FunH$. 
\end{fact}

Then we define coinductively the following relation:

\begin{definition}[Coinductive $\skipAct$-Compliance]
\[\complyGcok{k}  ~~\ByDef~~ \FunH^k~~~~~~ \mbox{and} ~~~~~~
\complyGco ~~\ByDef~~ \nu(\FunH),\]
where $\FunH^k$ is defined as in Fact  \ref{fact:gfpH}.

A {\em client} $\rho$ is said to be {\em coinductively} $\skipAct$-{\em compliant} with a {\em server} $\sigma$, whenever $\rho \complyGco \sigma$ holds.
\end{definition}

We say ``coinductively compliant" as short for ``coinductively $\skipAct$-compliant". By the last definition we have that
$\rec x.b.x\not\!\!\complyGco\rec x.\Dual{a}.x$. We have as well that $b\not\!\!\complyGco\rec x.(\Dual{a}.x \oplus \Dual{b})$.
In fact, $\Dual{a}\,\Dual{a}\,\Dual{a}...\in\trace(\rec x.(\Dual{a}.x \oplus \Dual{b}))$ and $\neg [\,b\,$ $\!\sync\!$ $({\rec x.(\Dual{a}.x \oplus \Dual{b})})]$.

\begin{proposition}\label{prop:coinductiveChar}
$$\complyG ~~=~~ \complyGco$$
\end{proposition}

\begin{proof}[sketch]
($\subseteq$)~
It sufficies to show that what stated in Definition \ref{def:coskipcompl}(\ref{def:coskipcompl-ii}) holds when we replace $\complyGco$ by $\complyG$. In case $\Converge{\rho}{\stopA}$,
we have that $\rho\complyG\sigma$ immediately by definition. Let us consider  the case
$\Converge{\rho} {\sum_{i\in I} a_i.{\rho}_i}$. Then we observe that $\rho \complyG \sigma$ and $\rho\not\Downarrow\stopA$
if and only if for any trace of $\sigma$ there exists a prefix $\vec{a}$ such that $\sigma\Lts{\vec{a}} \bigoplus_{h\in H}\Dual{a}_h.\sigma_h$ for some $H\subseteq I$ and $\rho_h\comply\sigma_h$ for all $h\in H$.
Moreover, $\exists k\in I.~\sync{a_k}{\sigma}$ holds since  $\forall k\in I.~\neg(\sync{a_k}{\sigma})$ contradicts Definition \ref{def:skipcompl}. 
If $\Converge{\rho} {\bigoplus_{i\in I} \Dual{a}_i.{\rho}_i}$ the proof proceeds in a similar way.

($\supseteq$)~ Let us assume $\rho\not\complyG\sigma$. This implies that $\neg(\Converge{\rho}{\stopA})$.
Then, by Definition \ref{def:skipcompl}, there exists $\vec{\xi}\in\cstrace(\rho \| \sigma)$ such that either $\vec{\xi}$ is finite but
 $\vec{\xi}\not=\vec{\xi}'\checkmark$ for any $\vec{\xi}'$, or 
$\vec{\xi}$ is infinite and definitely-$\skipAct$. In the first case we proceed by induction
on the lenght of the $\tau$-actions in $\vec{\xi}$ to contradict condition $ \exists k\in I.~\sync{a_k}{\sigma}$ in Definition \ref{def:coskipcompl}(\ref{def:coskipcompl-ii}).
In the infinite case, we get a contradiction to the \\
$\forall i\in I.\forall \sigma'.~[(\alpha_i.\rho_i \|\sigma \Lts{\skipAct^*\tau} \rho_i\|\sigma') ~\implies~ \rho_i\Rel \sigma']$ (for the proper $\alpha_i$) clauses in Definition \ref{def:coskipcompl}(\ref{def:coskipcompl-ii}).
\end{proof}

It is possible to show the relation $\complyG$ to be decidable. In order to do that we define a formal system that reflects the coinductive definition of the $\complyG$ relation, and whose 
derivation rules can be looked at as rules of a recursive, syntax-driven decision algorithm, where
the decision process coincides with a proof reconstruction procedure.\\

In the formal system, the assumptions in an environment are actually {\em marked} assumptions. The markings are used to prevent the possibility of getting a correct derivation for
compliance statements that allow for definitely-$\skipAct$ client$\|$server interactions.

\begin{definition}[A formal system for $\complyG$]\hfill
\label{def:formalSyst}
\begin{enumerate}[i)]
\item A marked environment $\Gamma$  is a finite set of marked assumptions of the form $(\rho' \complyG \sigma')_{\bullet}$, where $\rho',\sigma'\in\Sbehav$ and $\bullet \in \{\cmark,\xmark\}$.
\item
A judgment is an expression of the form $\Gamma\der \rho \complyG \sigma$, where 
$\Gamma$ is a marked environment. 
The axioms and rules of the system deriving judgments are in Figure \ref{fig:complianceSystem}, where the environment $\Gamma_{\cmark}$ is defined by $\Gamma_{\cmark} = \{ (\rho'\complyG\sigma')_{\cmark} | (\rho'\complyG\sigma')_{\cmark} \in \Gamma \vee (\rho'\complyG\sigma')_{\xmark}\in \Gamma \}$.
\end{enumerate}
\end{definition}
\noindent
We assume any marked environment to be {\em coherent}, that is there can be no
two assumptions with the same compliance statement and different markings in the same environment, like $(\rho \complyG \sigma)_{\cmark}$ and $(\rho \complyG \sigma)_{\xmark}$. Moreover, it will be easy to check that the derivation reconstruction procedure always produce coherent environments. 

\begin{figure}
\begin{center}
\line(1,0){450}
\end{center}
\[
\begin{array}{c@{\hspace{12mm}}c}
\prooftree
\justifies
\Gamma\der \stopA \complyG \sigma
\using(\ScomplAx)
\endprooftree &
\prooftree
\justifies
\Gamma, (\rho\complyG\sigma)_{\cmark} \der \rho\complyG\sigma 
\using(\ScomplHyp)
\endprooftree\\[8mm]
\prooftree \Gamma \der \sigma\Subst{\rec x. \sigma}{x} \complyG \sigma' \justifies
\Gamma \der \rec. \sigma \complyG  \sigma' \using (\ScomplUnfoldL) \endprooftree 
&
\prooftree \Gamma \der \sigma'  \complyG  \sigma\Subst{\rec x. \sigma}{x} \justifies
\Gamma \der \sigma'  \complyG  \rec x. \sigma  \using (\ScomplUnfoldR) \endprooftree 
\end{array}
\]
\vspace{3mm}
$$
\begin{array}{c}
\prooftree 
    \forall i\in (I\setminus K).~ \Gamma'\der \sum_{k\in K} a_k.{\rho}_k\complyG\sigma_i \qquad
    \forall j\in (K\cap I).~ \Gamma''\der \rho_j\complyG\sigma_j
\justifies
    \Gamma\der\sum_{k\in K} a_k.{\rho}_k \complyG \bigoplus_{i\in I} \Dual{a}_i.{\sigma}_i
\using(\ScomplSumOplus)
\endprooftree 
\\[8mm]
\mbox{ where } \Gamma' = \Gamma,\; (\sum_{k\in K} a_k.{\rho}_k \complyG \bigoplus_{i\in I} \Dual{a}_i.{\sigma}_i)_{\xmark}\\
\mbox{ ~~~~~~~~~~~~~~  }\Gamma'' = \Gamma_{\cmark},\; (\sum_{k\in K} a_k.{\rho}_k \complyG \bigoplus_{i\in I} \Dual{a}_i.{\sigma}_i)_{\cmark}
\end{array}
$$
$$
\begin{array}{c@{\hspace{12mm}}c}
\prooftree 
    \forall i\in I.~ \Gamma'\der \bigoplus_{k\in K} \Dual{a}_k.{\rho}_k\complyG\sigma_i 
\justifies
    \Gamma\der\bigoplus_{k\in K} \Dual{a}_k.{\rho}_k \complyG \bigoplus_{i\in I} \Dual{b}_i.{\sigma}_i
\using(\ScomplOplusOplus)
\endprooftree 
&
\prooftree 
    K\subseteq I \qquad \forall k\in K.~ \Gamma'\der \rho_k\complyG\sigma_k
\justifies
    \Gamma\der\bigoplus_{k\in K} \Dual{a}_k.{\rho}_k \complyG \sum_{i\in I} a_i.{\sigma}_i
\using(\ScomplOplusSum)
\endprooftree \\[8mm]
\mbox{ where } \Gamma' = \Gamma,\; (\bigoplus_{k\in K} \Dual{a}_k.{\rho}_k \complyG \bigoplus_{i\in I} \Dual{b}_i.{\sigma}_i)_{\xmark}
&
\mbox{ where } \Gamma' = \Gamma_{\cmark},\;(\bigoplus_{k\in K} \Dual{a}_k.{\rho}_k \complyG \sum_{i\in I} a_i.{\sigma}_i)_{\cmark}
\end{array}
$$
\caption{The formal system $\der$ for  $\complyG$}
\label{fig:complianceSystem}
\begin{center}
\line(1,0){450}
\end{center}
\end{figure}
\medskip
The intended meaning of a judgment $\Gamma\der \rho \complyG \sigma$ is
that if, for any $(\rho' \complyG \sigma')_\bullet \in \Gamma$, $\rho' \complyG \sigma'$ holds, then $\rho \complyG \sigma$ holds as well, except for some judgments for which the interaction between $\rho$ and $\sigma$ would produce definitely-$\skipAct$ syncrhonization traces. 
The use of markings rules out such a possiblity.
In fact the following Soundness and Completeness result we obtain is, as needed, for derivations with 
empty environment.

\bigskip

\begin{theorem}[Soundness and Completeness]\label{thr:completeness}
\[\rho\complyG\sigma ~~~~\Iff ~~~~ \emptyset\der \rho\complyG\sigma.\]
\end{theorem}

The proof of the Soundness and Completeness property above can be obtained by first 
proving it for a system without markings and allowing for definitely-$\skipAct$ interaction sequences in the definition of compliance. And then by showing that such definitely-$\skipAct$ sequences are ruled out if the derivations are properly marked.
The proof of the first part can be obtained along the lines used in \cite{BdL13} for a similar system. \\

\begin{theorem}[Decidability of $\der$]\label{thr:der-decidability}\label{rem:algorithmic}
The system of Figure \ref{fig:complianceSystem} is decidable.
\end{theorem}

\begin{proof}[Sketch]
The system in Figure \ref{fig:complianceSystem} satisfies a sort of subformula property. As a matter of fact the behaviours used in the premises of any rule are
subterms (for a suitable and natural definition of subterm) of those used in the premises of the rule. This implies the system to be algorithmic:
a decision procedure consists in a breadth first searching for a proof of a judgment in a bottom-up, syntax-driven, way. 
Such proof reconstruction ends since, for any possible branch of the proof,  we eventually find either $(1)$ an axiom $(\ScomplAx)$ or $(2)$ an hypothesis $(\ScomplHyp)$ or $(3)$ a wrong hypothesis, that is a judgment of the form $\Gamma, (\rho\complyG \sigma)_{\xmark} \der \rho\complyG\sigma$ or, by the subformula property, $(4)$ a previously encountered judgment. 
In case $(3)$ or $(4)$ are encountered along a branch, the derivation reconstruction algorithm
fails. In particular, the presence of $(3)$ denotes the possibility of a definitely-$\skipAct$ synchronization trace.\\
Notice that the proof reconstruction is also deterministic, but possibly for the choice
of the order in which $(\ScomplUnfoldL)$ and $(\ScomplUnfoldR)$ occur along a branch in the proof tree, which is immaterial as they have to be consecutive.
The complete proof develops along the same lines used for a similar proof in \cite{BdL13}, where we resort to a similar argument used in \cite{PierceSangiorgi95} and thereafter in \cite{GH05}.
\end{proof}

Decidability of compliance is now easily got  as a corollary.

\medskip

\begin{corollary}\label{cor:skip-compl-decidability}
The relation $\complyG$ is decidable.
\end{corollary}

\begin{proof} By Theorems \ref{thr:completeness} and \ref{rem:algorithmic}.
\end{proof}


\medskip
\noindent

\Comment{
\begin{example}\label{ex:GH-system}
Given $\rho = \rec x.(a.d+b.x)$ and $\sigma = \rec x.\Dual{c}.\Dual{b}.x$, we can derive:
$;\cmark \der \rho \complyG \sigma$ as follows:
\[
		\prooftree
			\prooftree
				\prooftree
					\prooftree
                                   \prooftree
					\justifies
						a.d + b.\rho \complyG \Dual{c}.\Dual{b}.\sigma,~ a.d + b.\rho \complyG \Dual{b}.\sigma ;\cmark  \der a.d + b.\rho \complyG \Dual{c}.\Dual{b}.\sigma
					\using (\ScomplHyp)
				\endprooftree
			\justifies
				a.d + b.\rho \complyG \Dual{c}.\Dual{b}.\sigma,~ a.d + b.\rho \complyG \Dual{b}.\sigma  ;\cmark \der  \rho \complyG \sigma 
			\using (\ScomplUnfoldLR)
			\endprooftree
		\justifies
			a.d + b.\rho \complyG \Dual{c}.\Dual{b}.\sigma  ;\xmark \der a.d + b.\rho \complyG \Dual{b}.\sigma
		\using (\ScomplSumOplus)
		\endprooftree
          \justifies
   ;\cmark  \der a.d + b.\rho \complyG \Dual{c}.\Dual{b}.\sigma
		\using (\ScomplSumOplus)
	\endprooftree
\justifies
	;\cmark \der \rho\complyG\sigma
\using (\ScomplUnfoldLR)
\endprooftree
\]
\end{example}
}

In the following we provide two simple example of application of the syntax-driven derivation reconstruction algorithm described in the proof of Theorem \ref{thr:der-decidability}. In the first example the algorithm fails because one of the possible interaction sequences for the client $b$ and the server $\rec x.(\Dual{a}.x \oplus \Dual{b})$ 
interaction would be definitely-$\skipAct$. In the second one the algorithm succeeds and produce the right derivation.
Notice how the failure in the first example is due to the fact that along the leftmost branch we 
encounter a {\em wrong hypothesis}, that is a judgment of the form $\Gamma, (\rho\complyG \sigma)_{\xmark} \der \rho\complyG\sigma$.

\begin{example}\label{counter-ex}
Given  $\sigma = \rec x.(\Dual{a}.x \oplus \Dual{b})$, the reconstruction algorithm for
$\der b \complyG \sigma$ produces the following result:
\[
				\prooftree
					\prooftree
\prooftree
	\begin{array}{c}
       \mbox{\em {\bf FAIL!}}\\
 (b \complyG  \Dual{a}.\sigma \oplus \Dual{b})_{\xmark}   \der b \complyG  \Dual{a}.\sigma \oplus \Dual{b}
       \end{array}
				\justifies
				 	(b \complyG  \Dual{a}.\sigma \oplus \Dual{b})_{\xmark} ~\der~ b \complyG \rec x.(\Dual{a}.x \oplus \Dual{b})
                          \using (\ScomplUnfoldR)
                       \endprooftree
                                                \qquad
                                                 \prooftree
                                                     \justifies
                                                     (b \complyG  \Dual{a}.\sigma \oplus \Dual{b})_{\cmark}  ~\der~ \stopA\complyG\stopA
                                                      \using (\ScomplAx)
                                                   \endprooftree
          \justifies
   ~\der~ b \complyG  \Dual{a}.\sigma \oplus \Dual{b}
		\using (\ScomplSumOplus)
		\endprooftree
	\justifies
		~\der~ b \complyG \rec x.(\Dual{a}.x \oplus \Dual{b})
	\using (\ScomplUnfoldR)
	\endprooftree
\]
\end{example}

\Comment{
\begin{example}\label{ex-inf-skip}
Given $\rho = \rec x.b.x$ and  $\sigma = \rec x.\Dual{a}.\Dual{a}.\Dual{b}.x$, we can derive:
$;\cmark \der \rho \complyG \sigma$ as follows:
\[
\prooftree
\prooftree
            \prooftree
                       \prooftree
                             \prooftree
                                    \prooftree
                                         \prooftree
                                         \justifies
                                        b.\rho  \complyG  \Dual{a}.\Dual{b}.\Dual{a}.\sigma, b.\rho \complyG \Dual{a}.\sigma  ;\xmark \der b.\rho  \complyG  \Dual{a}.\Dual{b}.\Dual{a}.\sigma
                                            \using (\ScomplHyp)
                                           \endprooftree
                                           \justifies
                                          b.\rho  \complyG  \Dual{a}.\Dual{b}.\Dual{a}.\sigma, b.\rho \complyG \Dual{a}.\sigma  ;\xmark \der b.\rho  \complyG  \rec x.\Dual{a}.\Dual{b}.\Dual{a}.x
                                          \using (\ScomplUnfoldR)
                                        \endprooftree
                                      \justifies
                            b.\rho  \complyG  \Dual{a}.\Dual{b}.\Dual{a}.\sigma;\cmark \der b.\rho \complyG  \Dual{a}.\sigma
                                       \using (\ScomplSumOplus)
                                       \endprooftree
                              \justifies
                             b.\rho  \complyG  \Dual{a}.\Dual{b}.\Dual{a}.\sigma  ;\cmark \der \rec x.b.x \complyG \Dual{a}.\sigma
                              \using (\ScomplUnfoldL)
                            \endprooftree
				 \justifies
				 b.\rho  \complyG  \Dual{a}.\Dual{b}.\Dual{a}.\sigma  ;\xmark \der b.\rho \complyG \Dual{b}.\Dual{a}.\sigma
                            \using (\ScomplSumOplus)
                        \endprooftree
               \justifies
                  ;\cmark  \der b.\rho  \complyG  \Dual{a}.\Dual{b}.\Dual{a}.\sigma
\using  (\ScomplSumOplus)
\endprooftree
\justifies
;\cmark  \der \rec x.b.x  \complyG  \rec x.\Dual{a}.\Dual{b}.\Dual{a}.x
\using (\ScomplUnfoldL) (\ScomplUnfoldR)
\endprooftree
\]
\end{example}

\begin{example}\label{ex-inf-skipNew}
Given $\rho = \rec x.b.x$ and  $\sigma = \rec x.\Dual{a}.\Dual{a}.\Dual{b}.x$, we can derive:
$\der \rho \complyG \sigma$ as follows:
\[
\prooftree
\prooftree
            \prooftree
                       \prooftree
                             \prooftree
                                    \prooftree
                                         \prooftree
                                         \justifies
                                       (b.\rho  \complyG  \Dual{a}.\Dual{b}.\Dual{a}.\sigma)_{\cmark}, (b.\rho \complyG \Dual{b}.\Dual{a}.\sigma)_{\cmark}, (b.\rho \complyG \Dual{a}.\sigma)_{\xmark} ~\der~ b.\rho  \complyG  \Dual{a}.\Dual{b}.\Dual{a}.\sigma
                                            \using (\ScomplHyp)
                                           \endprooftree
                                           \justifies
                                           (b.\rho  \complyG  \Dual{a}.\Dual{b}.\Dual{a}.\sigma)_{\cmark}, (b.\rho \complyG \Dual{b}.\Dual{a}.\sigma)_{\cmark}, (b.\rho \complyG \Dual{a}.\sigma)_{\xmark} ~\der~ b.\rho  \complyG  \rec x.\Dual{a}.\Dual{b}.\Dual{a}.x
                                          \using (\ScomplUnfoldR)
                                        \endprooftree
                                      \justifies
                             (b.\rho  \complyG  \Dual{a}.\Dual{b}.\Dual{a}.\sigma)_{\cmark}, (b.\rho \complyG \Dual{b}.\Dual{a}.\sigma)_{\cmark} ~\der~ b.\rho \complyG  \Dual{a}.\sigma
                                       \using (\ScomplSumOplus)
                                       \endprooftree
                              \justifies
                             (b.\rho  \complyG  \Dual{a}.\Dual{b}.\Dual{a}.\sigma)_{\cmark}, (b.\rho \complyG \Dual{b}.\Dual{a}.\sigma)_{\cmark}  ~\der~ \rec x.b.x \complyG \Dual{a}.\sigma
                              \using (\ScomplUnfoldL)
                            \endprooftree
				 \justifies
				 (b.\rho  \complyG  \Dual{a}.\Dual{b}.\Dual{a}.\sigma)_{\xmark} ~\der~ b.\rho \complyG \Dual{b}.\Dual{a}.\sigma
                            \using (\ScomplSumOplus)
                        \endprooftree
               \justifies
                 ~\der~ b.\rho  \complyG  \Dual{a}.\Dual{b}.\Dual{a}.\sigma
\using  (\ScomplSumOplus)
\endprooftree
\justifies
~\der~ \rec x.b.x  \complyG  \rec x.\Dual{a}.\Dual{b}.\Dual{a}.x
\using (\ScomplUnfoldL) (\ScomplUnfoldR)
\endprooftree
\]
\end{example}
}

\begin{example}\label{ex-inf-skipNew2}
Given $\rho = \rec x.b.x$ and  $\sigma = \rec y.(\Dual{a}.\Dual{b}.\Dual{a}.y \oplus \Dual{b}.\rec x.\Dual{b}.x)$, the reconstruction algorithm for
$\der \rho \complyG \sigma$ produces the following result:
\[
\prooftree
\prooftree
\prooftree
\prooftree 
            \prooftree
                       \prooftree
                             \prooftree
                                    \prooftree
                                         \prooftree
                                         \justifies
                                       (\gamma_1)_{\cmark}, (\gamma_2)_{\cmark}, (\gamma_3)_{\cmark},(\gamma_4)_{\xmark} ~\der~  \gamma_2 %
                                            \using (\ScomplHyp)
                                           \endprooftree
                                           \justifies
                                           (\gamma_1)_{\cmark}, (\gamma_2)_{\cmark}, (\gamma_3)_{\cmark},(\gamma_4)_{\xmark} ~\der~ b.\rho  \complyG  \sigma%
                                          \using (\ScomplUnfoldR)
                                        \endprooftree
                                      \justifies
                              (\gamma_1)_{\cmark}, (\gamma_2)_{\cmark}, (\gamma_3)_{\cmark} ~\der~ b.\rho \complyG  \Dual{a}.\sigma%
                                       \using (\ScomplSumOplus)
                                       \endprooftree
                              \justifies
                             (\gamma_1)_{\cmark}, (\gamma_2)_{\cmark}, (\gamma_3)_{\cmark} ~\der~ \rho \complyG \Dual{a}.\sigma%
                              \using (\ScomplUnfoldL)
                            \endprooftree
				 \justifies
				 (\gamma_1)_{\xmark}, (\gamma_2)_{\xmark}  ~\der~ b.\rho \complyG \Dual{b}.\Dual{a}.\sigma %
                            \using (\ScomplSumOplus)
                        \endprooftree 
                        \qquad
                     \prooftree        
                          \prooftree
                                \prooftree
                                    \prooftree
                                      \justifies
                                  (\gamma_1)_{\cmark}, (\gamma_2)_{\cmark}, (\gamma_5)_{\cmark} ~\der~\gamma_5
                                       \using (\ScomplHyp)
                                     \endprooftree
                                \justifies
                              (\gamma_1)_{\cmark}, (\gamma_2)_{\cmark}, (\gamma_5)_{\cmark} ~\der~ \rho  \complyG  \rho
                                 \using(\ScomplUnfoldLR)
                                \endprooftree
                             \justifies
                        (\gamma_1)_{\xmark}, (\gamma_2)_{\xmark} ~\der~ b.\rho  \complyG  \Dual{b}.\rho
                          \using (\ScomplSumOplus)
                          \endprooftree
                        \justifies
           (\gamma_1)_{\xmark}, (\gamma_2)_{\xmark} ~\der~ \rho  \complyG  \rec x.\Dual{b}.x
                          \using(\ScomplUnfoldLR)
                     \endprooftree   
               \justifies
                 (\gamma_1)_{\xmark} ~\der~ b.\rho  \complyG  \Dual{a}.\Dual{b}.\Dual{a}.\sigma \oplus \Dual{b}.\rec x.\Dual{b}.x) %
\using  (\ScomplSumOplus)
\endprooftree
\justifies
(\gamma_1)_{\xmark}  ~\der~ b.\rho  \complyG  \sigma %
\using (\ScomplUnfoldR)
\endprooftree
\justifies
~\der~ b.\rho  \complyG  \Dual{c}.\sigma %
\using (\ScomplSumOplus)
\endprooftree
\justifies
~\der~ \rec x.b.x  \complyG  \Dual{c}.\rec y.(\Dual{a}.\Dual{b}.\Dual{a}.y \oplus \Dual{b}.\rec x.\Dual{b}.x)
\using (\ScomplUnfoldL)
\endprooftree
\]
$
\begin{array}{ll@{\hspace{12mm}}l}
\mbox{where } & \gamma_1 = b.\rho  \complyG  \Dual{c}.\sigma 
                    & \gamma_4 = b.\rho \complyG  \Dual{a}.\sigma\\
    & \gamma_2 =b.\rho  \complyG  \Dual{a}.\Dual{b}.\Dual{a}.\sigma \oplus \Dual{b}.\rec x.\Dual{b}.x & \gamma_5 = b.\rho  \complyG  \Dual{b}.\rho \\
    &  \gamma_3 = b.\rho \complyG \Dual{b}.\Dual{a}.\sigma
\end{array}
$
\end{example}

\section{The $\skipAct$-subbehaviour relation}
\label{sect:skipSub}

As mentioned in the Introduction, in the theory of contracts the compliance relation induces a preorder $\preceq$. The relation $\sigma\preceq \sigma'$ holds whenever, for any client $\rho$,  
if $\rho\comply\sigma$ then $\rho\comply\sigma'$ . 

If $\sigma$, $\sigma'$ and $\rho$ are required to be in $\Sbehav$ then this relation,
which we call {\em subbehaviour relation} (dubbed $\preceq_s$ in \cite{BdL10}),
 coincides with the testing must-preorder \cite{DeNicolaH83}, which is not the
case if arbitrary contracts are considered (see \cite{BH13}). Here we relativize the definition of the subbehavior relation to the $\complyG$ relation
studied in the previous section, obtaining a new relation, which we call {\em  $\skipAct$-subbehaviour} and dub $\preceqG$.


\begin{definition}[$\skipAct$-Subbehaviour]\label{def:skip-subbehaviour}
~\\
Over $\Sbehav$ it is defined the binary relation $\sigma\preceqG \sigma'$ by
\[\sigma\preceqG \sigma' ~~\Iff~~\forall \rho\in\Sbehav. ~[\rho\complyG\sigma \implies \rho\complyG\sigma'].\]
\end{definition}


\begin{remark}{\em 
It is not difficult to check that $\preceqG \not\subseteq\preceq$ by means of the following easy counterexample. We have that $a\preceqG\Dual{c}.a$. In fact all the possible $\skipAct$-compliant clients of $a$ are $\{\stopA,\Dual{a}\}$, which are trivially also $\skipAct$-compliant with the server $\Dual{c}.a$ by skipping the action
$\Dual{c}$. Without the possibilty of skipping such an action,
we have that $\Dual{a}$ is not a client of $\Dual{c}.a$ anymore, whereas it is still so of $a$.
That is $a\not\preceq\Dual{c}.a$.\\
For what concern the opposite inclusion, we conjecture it to hold. Of course the proof would not be immediate. 
By Remark \ref{rem:strongComplInclWeak}, we know that $\comply \; \subset \; \complyG$, but this fact doesn't imply that $\preceq$ is included
in $\preceqG$, because $\sigma\preceqG \sigma'$ depends on a negative occurrence of the hypothesis $\rho\complyG\sigma$.
}\end{remark}

\Comment{
\marginpar{il rem. contiene solo abbozzi}
\begin{remark}{\em Are $\preceq$ and $\preceqG$ related w.r.t. set inclusion?
By Remark \ref{rem:strongComplInclWeak} we know that $\comply \; \subset \; \complyG$. This fact doesn't imply that $\preceq$ is included
in $\preceqG$, because $\sigma\preceqG \sigma'$ depends on a negative occurrence of the hypothesis $\rho\complyG\sigma$.

Behaviors which are suffix one of the other are not related w.r.t. $\preceqG$ in general. E.g. $\Dual{b}\not\preceqG \Dual{a}.\Dual{b}$ because
$a.c + b \complyG \Dual{b}$ but $a.c + b \not\complyG  \Dual{a}.\Dual{b}$. Also $\Dual{a}.\Dual{b} \not\preceqG \Dual{b}$ because
$a \complyG\Dual{a}.\Dual{b}$ but $a \not\complyG \Dual{b}$. Note that, however, $\Dual{b}\not\preceq \Dual{a}.\Dual{b}$ and $\Dual{a}.\Dual{b} \not\preceq \Dual{b}$.

In the case of behaviours which are one the prefix of the other the relations 
$\preceq$ and $\preceqG$ seem to coincide: $\Dual{a} \preceqG \Dual{a}.\Dual{b}$, but also
$\Dual{a} \preceq \Dual{a}.\Dual{b}$. 
Also the covariant law of $\oplus$ seems to be preserved from $\preceq$ to $\preceqG$: $\Dual{a} \oplus \Dual{b} \preceqG \Dual{a}$.

A possibility (and a conjecture) is that, in spite servers have more clients under $\complyG$ than under $\comply$, the preorders $\preceq$ and $\preceqG$ coincide.
This is the case of $\Dual{a}\oplus\Dual{b}.\Dual{c} \preceqG \Dual{a}$ which also holds w.r.t. $\preceq$, but $a+c \complyG \Dual{a}\oplus\Dual{b}.\Dual{c}$ (and $a+c\complyG \Dual{a}$) although  $a+c \not\comply \Dual{a}\oplus\Dual{b}.\Dual{c}$.
}\end{remark}
}

\subsection{{\em Duals as minima} property and decidability of $\preceqG$}
We proceed now towards the proof of decidability of the $\skipAct$-subbehaviour relation.
This will be obtained as a corollary of the property that the dual of a session-behaviour is actually the minimum among its servers
w.r.t. $\preceqG$. 
For any theory of  subcontracts this {\em duals as minima} result is quite relevant,
since the possibility of implementing contract-based query engines relies on it.
This is well explained in \cite{Padovani10} in the paragraph that we quote below.
\noindent
\begin{quote}
\small 
Formal notions of compliance and subcontract relation may be used for implementing
contract-based query engines. The query for services that satisfy
$\rho$ is answered with the set ${\cal Q}_1 (\rho) = \{\sigma\mid \rho\comply\sigma\}$. The complexity of running
this query grows with the number of services stored in the repository. A better
strategy is to compute the dual contract of $\rho$, denoted by $\rho^\perp$ [$\Dual{\rho}$ {\em in our context}], which represents
the canonical service satisfying $\rho$ (that is $\rho\comply\rho^\perp$) and then answering the
query with the set ${\cal Q}_2  (\rho) = \{\sigma\mid \sigma\preceq\rho^\perp\}$. If $\rho^\perp$ is the $\preceq$-smallest service that
satisfies $\rho$, we have ${\cal Q}_1 (\rho) = {\cal Q}_2  (\rho)$, namely we are guaranteed that no service
is mistakenly excluded. The advantage of this approach is that $\preceq$ can be precomputed
when services are registered in the repository, and the query engine
needs only scan through the $\preceq$-minimal contracts. \\
({\em L.Padovani} - \cite{Padovani10}, Sect.1)
\end{quote}

\noindent
The minimum property of dual behaviours can be proved using the following property: \vspace{-2mm}
\begin{equation}\label{eq:minAux}
\rho\complyG\gamma \And \Dual{\gamma}\complyG\sigma ~~\implies~~ \rho\complyG\sigma
\end{equation}
This property, however, is not easy to establish in presence of skipped actions, as exemplified in the following.
It is immediate to check that:
\vspace{-2mm}
$$a.d\complyG\Dual{b}.\Dual{b}.\Dual{a}.\Dual{d}~~~~\And ~~~~b.b.a.d\complyG\Dual{a}.\Dual{b}.\Dual{a}.\Dual{b}.\Dual{a}.\Dual{d}$$
Each $b$ in $b.b.a.d$ skips an $\Dual{a}$ before synchronizing with its dual $\Dual{b}$,
whereas the action that synchronizes with the $a$ in $b.b.a.d$ is the last $\Dual{a}$ of 
$\Dual{a}.\Dual{b}.\Dual{a}.\Dual{b}.\Dual{a}.\Dual{d}$.
Now, the $a$ in $a.d$ synchronizes after skipping the two $\Dual{b}$'s corresponding to the 
first two $b$'s of $b.b.a.d$. The action in $\Dual{a}.\Dual{b}.\Dual{a}.\Dual{b}.\Dual{a}.\Dual{d}$ synchronizing with the $a$ in $a.d$, however, it is not the last $\Dual{a}$ of  $\Dual{a}.\Dual{b}.\Dual{a}.\Dual{b}.\Dual{a}.\Dual{d}$, but actually the first one.

This fact, fortunately, does not cause any problem for
$a.d\complyG\Dual{a}.\Dual{b}.\Dual{a}.\Dual{b}.\Dual{a}.\Dual{d}$
since the $d$ in $a.d$ synchronizes with the $\Dual{d}$ of $\Dual{a}.\Dual{b}.\Dual{a}.\Dual{b}.\Dual{a}.\Dual{d}$ by skipping all the actions $\Dual{b}.\Dual{a}.\Dual{b}.\Dual{a}$ between the first $\Dual{a}$ and $\Dual{d}$. 
The presence of cases like these require to be carefully handled when proving property 
(\ref{eq:minAux}) that is otherwise similar to the analogous facts in \cite{BdL10,BdL13}.

To ease the proof we first consider an equivalent formulation of the $\skipAct$-compliance relation.
We introduce a relation $\substr$ between sequences of actions, such that $a_1\ldots a_n \substr b_1\ldots b_m$ holds whenever any $a_i$ (going from 
left to right) coincides with some $b_j$, provided that all the elements between the element $b_h$ coinciding with $a_{i-1}$ and $b_j$ are distinct from $b_j$. For instance, $bbad \substr abababd$  and $ad \substr abababd$, whereas $ad \not\substr bbaa$.

\begin{definition}[The $\substr$ relation.]\hfill
\begin{enumerate}[i)]
\item
The binary relation $\substr\;\,\subseteq\Names^+\!\times\Names^+$ on finite and non empty sequences of input actions  is inductively defined as follows.\\
Let $\vec{a},\vec{b}\in\Names^+$. 
\begin{itemize}
\item
$b \substr a_1\ldots a_kb  ~~\ByDef~~ k\geq 0 \And b \not= a_1,\ldots, a_k$
\vspace{2mm}
\item
$b\,\vec{a} \substr a_1\ldots a_k b\,\vec{b}~~\ByDef~~\vec{a}\substr \vec{b} \And k\geq 0 \And b \not= a_1,\ldots, a_k$
\end{itemize}
\vspace{2mm}
\item
The above relation is naturally extended to $\Names^\infty\times\Names^\infty$ and to $\Names^+\times\Names^\infty$
\end{enumerate}
\end{definition}

The relation $\substr$ will be used in the alternative coinductive $\skipAct$-compliance provided in
Lemma \ref{lem:AlternativacoinductiveSkipCompliance} below. It will be used to represent, on the left-hand side, the
synchronizing actions of the client and on the right-hand side the corresponding actions of the server, possibly preceded by a finite number of {\em skipped} actions. The relation is extended to  $\Names^\infty\times\Names^\infty$ since a client can be $\skipAct$-compliant with a server even without ever terminating. It is extended to $\Names^+\times\Names^\infty$ since a client can
succesfully terminate even if its server could be able to going on indefinitely. 

The following property holds for $\substr$.
\begin{lemma}
\label{lem:transSubstr}
The relation $\substr$ is transitive.
\end{lemma}

\begin{lemma}[Alternative coinductive $\skipAct$-compliance]\label{lem:AlternativacoinductiveSkipCompliance}

$$\nu(\FunH) = \nu(\FunJ)$$
where 
the operator 
$\FunJ: {\cal P}(\Sbehav\times \Sbehav)\rightarrow  {\cal P}(\Sbehav\times \Sbehav)$
is defined as follows: 
 for any relation $\Rel\subseteq \Sbehav\times \Sbehav$, 
$(\rho,\sigma) \in {\FunJ}(\Rel)$ if and only if  either $\Converge{\rho}{\stopA}$ or,
whenever $[\notConverge{\rho} {\sum_{i\in I} a_i.{\rho}_i} ~\&~ \notConverge{\sigma} {\sum_{j\in J} a_j.{\sigma}_j} ]$, the following statements hold: 
\vspace {1mm}


\begin{enumerate}[a)]
\item $\Converge{\rho} {\sum_{i\in I} a_i.{\rho}_i} ~~\Rightarrow~~
           \left\{\begin{array}{l}
           \{\Dual{\vec{b}} \mid |\vec{b}|>0, ~\sigma\Lts{\Dual{\vec{b}}} \} \not= \emptyset \\[3mm]
           \forall \Dual{\vec{b}}{~s.t.~} \sigma\Lts{\Dual{\vec{b}}}\sigma'.~ 
         ~\exists \vec{a}\substr\vec{b}.~~ (\rho\LtsM{\vec{a}}\rho' 
            \And \rho'\FunJ\sigma' )\\[3mm]
           \forall \Dual{\vec{b}}\,^\infty\in\trace(\sigma).
         ~\exists \vec{a}\substr\vec{b}.~~(\rho\LtsM{\vec{a}}\stopA 
          \Or \rho\LtsM{\vec{a}^\infty})
           \end{array}
             \right.$

         \vspace{3mm}

\item $\Converge{\rho} {\bigoplus_{i\in I} \Dual{a}_i.{\rho}_i} ~~\Rightarrow~~
             \left\{\begin{array}{l}
                       \{\Dual{\vec{a}}^{\;\infty} \mid \sigma\Lts{\Dual{\vec{a}}^{\;\infty}} \} = \emptyset \\[1mm]
                         \forall j\in I~.\forall \vec{a}~s.t.~ \sigma\LtsM{\Dual{\vec{a}}}~.~(\sigma\Lts{\Dual{\vec{a}}\,a_j}\sigma' \And
                       \rho_j \FunJ  \sigma')
                 \end{array}
             \right.$
\end{enumerate}
where
$\sigma\LtsM{\vec{a}} ~\ByDef~~ [\sigma\Lts{\vec{a}} \And \not\exists \,c\in\Names.\,\sigma\Lts{\vec{a}\,c}]$ ~~and~~
$\sigma\LtsM{\Dual{\vec{a}}} ~\ByDef~~ [\sigma\Lts{\Dual{\vec{a}}} \And \not\exists \,\Dual{c}\in\CoNames.\,\sigma\Lts{\Dual{\vec{a}}\,\Dual{c}}]$.   

\end{lemma}


The following property will be useful to show the dual-as-minimum property.

\begin{lemma}\label{lem:aux}
Given $\sigma\LtsM{\Dual{\vec{a}}}$ with $\gamma\complyG\sigma$,
there exists $\vec{c}$ s.t. ~ $\gamma\LtsM{\vec{c}}\gamma' \And \Converge{\gamma'}{\bigoplus \Dual{b}_j.\gamma'_{j}}$.
Moreover, for any $\Dual{b}_j$ we have $\sigma\Lts{\Dual{\vec{a}}\;b_j}\sigma'_j$ with 
$\gamma'_j\complyG\sigma'_j$.
\end{lemma}

\begin{lemma}\label{lemJ:transUpDuality}
For all $\rho,\sigma,\gamma\in\Sbehav$: ~
if $\rho\complyG\gamma$ and $\Dual{\gamma}\complyG\sigma$ then $\rho\complyG\sigma$.
\end{lemma}

\begin{proof}
By Lemma \ref{prop:coinductiveChar}, we have to prove
the relation 
$\RelKK= \{ (\rho,\sigma) \mid \exists \gamma. ~ \rho\complyG\gamma \And \Dual{\gamma}\complyG\sigma\}$
to be a coinductive Skip-compliance. We shall do that by using the alternative characterization of coinductive $\skipAct$-compliance of Lemma \ref{lem:AlternativacoinductiveSkipCompliance}.
Let $\rho$ and $\sigma$ be such that $\rho\complyG\gamma$ and $\Dual{\gamma}\complyG\sigma$ for some $\gamma$. 
There are two cases, of which we consider the most complex one for lack of space.
\begin{description}
\Comment{
\item $\Converge{\rho} {\sum_{i\in I} a_i.{\rho}_i}$\;\; :\;
We know that there exists $\gamma$ s.t. $\rho\complyG\gamma$ and  $\Dual{\gamma}\complyG\sigma$. Since $\gamma$ cannot be an external choice, we have that, necessarily, $\Converge{\Dual{\gamma}} {\sum_{h\in H} c_h.{\Dual{\gamma}}_h}$.
We need to distinguish whether a maximal sequence of actions out of $\sigma$ be finite or infinite. 
\begin{description}
\item $\sigma\LtsM{\Dual{\vec{b}}}\sigma'$\\
Since $\Dual{\gamma}\complyG\sigma$, by the characterization of Lemma \ref{lem:AlternativacoinductiveSkipCompliance}, we get that there exists
$\vec{c}\substr \vec{b}$  such that
$\Dual{\gamma}\LtsM{\vec{c}}\Dual{\gamma}'$,
with $\Dual{\gamma}'\complyG \sigma'$.

 By duality we have that 
$\gamma\LtsM{\Dual{\vec{c}}}\gamma'$. 
Now, from $\rho\complyG\gamma$ and Lemma \ref{lem:AlternativacoinductiveSkipCompliance} again,
 we get that there exists $\vec{a}\substr \vec{c}$ such that $\rho\LtsM{\vec{a}}\rho'$,
with $\rho'\complyG \gamma'$.

By Lemma \ref{lem:transSubstr}, from $\vec{a}\substr \vec{c}$ and $\vec{c}\substr \vec{b}$ we get 
$\vec{a}\substr \vec{b}$. Now, from $\rho'\complyG \gamma'$ and $\Dual{\gamma}'\complyG \sigma'$, we get $\rho'\RelKK\sigma'$, and hence the thesis.\\
\item 
$\sigma\LtsM{\Dual{\vec{b}}\,^\infty}$\\
By Lemma \ref{lem:AlternativacoinductiveSkipCompliance} and $\Dual{\gamma}\complyG\sigma$,
there exists $\vec{c}\substr\vec{b}$ such that either $\Dual{\gamma}\LtsM{\vec{c}}\stopA$ or 
$\Dual{\gamma}\LtsM{\vec{c}^\infty}$.

In the first case, by duality, we have $\gamma\LtsM{\Dual{\vec{c}}}\stopA$,
and hence from $\rho\complyG\gamma$ and  Lemma \ref{lem:AlternativacoinductiveSkipCompliance},
there exists $\vec{a}\substr \vec{c}$ such that $\rho\LtsM{\vec{a}}\rho'$,
with $\rho'\complyG \stopA$. Since $\rho'\complyG \stopA$ implies $\rho'\equiv\stopA$ and hence $\rho\LtsM{\vec{a}}\stopA$, from which the thesis, since we have also 
$\vec{a}\substr \vec{b}$ from $\vec{a}\substr \vec{c}$ and $\vec{c}\substr \vec{b}$. 

In the second case, by duality, we have $\gamma\LtsM{\Dual{\vec{c}}\,^\infty}$,
and hence from $\rho\complyG\gamma$ and  Lemma \ref{lem:AlternativacoinductiveSkipCompliance},
there exists $\vec{a}\substr\vec{c}$ such that either $\rho\LtsM{\vec{a}}\stopA$ or 
$\rho\LtsM{\vec{a}^\infty}$. 
Since $\vec{a}\substr \vec{b}$ from $\vec{a}\substr \vec{c}$ and $\vec{c}\substr \vec{b}$
we get immediately the thesis in the first case, whereas, in the second case, we get the thesis
since, by Lemma \ref{lem:transSubstr}, $\vec{a}^\infty\substr \vec{b}^\infty$ from $\vec{a}^\infty\substr \vec{c}^\infty$ and $\vec{c}^\infty\substr \vec{b}^\infty$.\\
\end{description}
}
\item $\Converge{\rho} {\bigoplus_{i\in I} \Dual{a}_i.{\rho}_i}$\;\; :\;
Let $k\in I$.
From $\rho\complyG\gamma$, \ref{lem:AlternativacoinductiveSkipCompliance}, we get that $\{\Dual{\vec{c}}^{\;\infty} \mid \gamma\Lts{\Dual{\vec{c}}^{\;\infty}} \} = \emptyset$ and that, 
\vspace{-3mm}
\begin{equation}\label{eq:i}
\forall\;\vec{\Dual{c}}~ s.t.~
\gamma\LtsM{\Dual{\vec{c}}}{\gamma''}~.  ~~\Converge{\gamma''}{\sum b_j.
\gamma'_j} \And \exists h.\; b_h\equiv a_k \And \rho_k\complyG\gamma'_h.
\end{equation}

\vspace{-4mm}
 By duality, we get that for all $\vec{c}$ s.t. $\Dual{\gamma}\LtsM{\vec{c}}\Dual{\gamma''}$, $\Converge{\Dual{\gamma''}}{\bigoplus \Dual{b}_j.\Dual{\gamma'}_j}$.
We can now infer that $\{\Dual{\vec{a}}^{\;\infty} \mid \sigma\Lts{\Dual{\vec{a}}^{\;\infty}} \} = \emptyset$, by distinguishing two cases: if $\Converge{\gamma}{\bigoplus \Dual{c}_p.{\gamma}_p}$, it is immediate by $\Dual{\gamma}\complyG\sigma$ and Lemma \ref{lem:AlternativacoinductiveSkipCompliance}.
Otherwise, by contradiction, let assume that there exists $\vec{a}^{\;\infty}$ such that $\sigma\Lts{\Dual{\vec{a}}^{\;\infty}}$. By $\Dual{\gamma}\complyG\sigma$ and 
Lemma \ref{lem:AlternativacoinductiveSkipCompliance} we get that there exists 
$\vec{d}\substr\vec{a}$ such that either $\Dual{\gamma}\LtsM{\vec{d}}$ or 
$\Dual{\gamma}\LtsM{\vec{d}^\infty}$. We obtain an immediate contradiction in the second case, whereas in the first one, we get a contradiction by the fact that 
$\gamma\LtsM{\Dual{\vec{d}}}\stopA$ and by (\ref{eq:i}).
Now, given $\sigma\LtsM{\Dual{\vec{a}}}$, from $\Dual{\gamma}\complyG\sigma$ and 
Lemma \ref{lem:aux}, given $\sigma\LtsM{\Dual{\vec{a}}}$, 
there exists $\vec{c}'$ s.t. $\Dual{\gamma}\LtsM{\vec{c'}}\Dual{\gamma_1''}$, $\Converge{\Dual{\gamma_1''}}{\bigoplus \Dual{b}_j.\Dual{\gamma'_1}_j}$,
moreover, for any $\Dual{b}_j$ we have $\sigma\Lts{\Dual{\vec{a}}\;b_j}\sigma'_j$ with 
$\Dual{\gamma'_1}_j\complyG\sigma'_j$.
From (\ref{eq:i}) we get that 
$$\gamma\LtsM{\Dual{\vec{c'}}}{\gamma''}~.  ~~\Converge{\gamma_1''}{\sum b_j.
{\gamma_1}'_j} \And \exists h.\; b_h\equiv a_k \And \rho_k\complyG{\gamma_1}'_h.$$
Since $\Dual{\gamma'_1}_h\complyG\sigma'_h$, we get $\rho_k\RelKK\sigma'_h$.
\end{description}
\end{proof}
\pagebreak
\begin{proposition}[Duals as minima]\label{prop:mins}\hfill\\
Let $\rho \in\Sbehav$. Then
$\Dual{\rho}$ is the minimum server of $\rho$, i.e. :~~~~
$\forall \sigma.\; ~\rho\complyG \sigma ~~\implies~~ \Dual{\rho} \preceqG \sigma$
\Comment{
\label{prop:mins-ii}
Let $\sigma \in\Sbehav$. Then
$\Dual{\sigma}$ is the minimum client of $\sigma$, i.e.
$$\forall \rho.\; ~\rho\complyG \sigma ~~\implies~~ \Dual{\sigma} \preceqG_c \rho$$
}
\end{proposition}

\begin{proof}
Let $\sigma$ and $\gamma$ be such that $\rho\complyG \sigma$ and 
$\gamma \complyG \Dual{\rho}$.
It is immediate to check that $\Dual{\rho}\complyG\rho$. Hence, by Lemma \ref{lemJ:transUpDuality}
and the fact that the $\Dual{\cdot}$ operation is involutive, we have that $\gamma \complyG \sigma$, so showing that $\Dual{\rho} \preceqG \sigma$.
\end{proof}
\Comment{
(\ref{prop:mins-ii})
Let $\rho$ and $\gamma$ be such that $\rho\complyG \sigma$ and 
$ \Dual{\sigma}\complyG \gamma$.
Hence, by Lemma \ref{lemJ:transUpDuality}, we have that $\rho \complyG \gamma$, so showing that $\Dual{\sigma} \preceqG_c \rho$.
}

We are finally in place to establish the following result.

\begin{theorem}\label{thm:dualsAsMinima}
\label{th:dualasminimum}
$~~~~~~~~\sigma\preceqG\sigma' ~~~~\iff~~~~ \Dual{\sigma}\complyG \sigma'$
\end{theorem}

\begin{proof}
($\Rightarrow$) Let $\Dual{\sigma}\not\complyG \sigma'$. Since we have $\Dual{\sigma}\complyG \sigma$, we get then that $\sigma\not\preceqG\sigma'$.

($\Leftarrow$) Let $\Dual{\sigma}\complyG \sigma'$. Then, by Proposition \ref{prop:mins}, we get
$\sigma=\Dual{\Dual{\sigma}}\preceqG\sigma'$.
\end{proof}

By  Theorem \ref{thm:dualsAsMinima} and decidability of $\complyG$ stated in Corollary \ref{cor:skip-compl-decidability} we conclude:

\begin{corollary}\label{cor:skip-preceq-decidability}
The relation $\preceqG$ is decidable.
\end{corollary}

\section{Related works}
\label{sect:skipRelated}

What we devised in the present paper is not the only possibility of weakening the notions of
compliance and sub-behaviour. An alternative approach in the setting of
(first-order and unrestricted) contracts has been followed  by Luca Padovani in \cite{Padovani10}.

We briefly recall Padovani's approach to compare with ours, which is possible because session-behaviours are particular contracts.
In \cite{Padovani10} the interactions between a client and a server can be mediated (coordinated) by
an {\em orchestrator}, a particular process (a  sort of {\em active channel} or {\em channel controller}) with the capability of buffering messages.
Thanks to that, the server's ``answers" to the client's  ``requests" can be delivered in a different order, so enabling a form of asynchronous interactions, or kept indefinitely in the buffer, that is equivalent to {\em discarding} them.
 The weak-compliance
relation resulting from the presence of orchestrators, that we denote here by $\complyP$, induces a preorder that is also investigated in \cite{Padovani10}, and that here we refer to as $\preceqP$.

Let us explain the use of orchestrators by means of an example.
The following is the behaviour of a ballot service similar to one we already 
described in the Introduction. Logging-in can be retried in case of wrong login or system overload. The message ${\tt Id}$ denotes 
the identifier of the transaction provided by the server to its clients.

\begin{tabbing}
${\sf BallotServic} \ByDef$ \= $\,\,\rec x.\;{\tt Login}.(\Dual{\tt Wrong}.\Dual{\tt InfoW}.x\;\oplus\; \Dual{\tt Overload}.x \;\oplus\;\,$\= \kill

${\sf BallotServiceBehP} \ByDef $\\
   \> $\rec x.\,{\tt Login}.(\Dual{\tt Wrong}.\,x\;\oplus\; \Dual{\tt Overload}.\,x \;\oplus\; \Dual{\tt Ok}.\Dual{\tt Id}.(\;\;\;\;{\tt VoteA}.({\tt Va1} + {\tt Va2})$\\
             \>                                  \>\,$+\; {\tt VoteB}.({\tt Vb1} + {\tt Vb2})\;)\;)$
\end{tabbing}
\Comment{
\begin{tabbing}
${\sf BallotServiceBehP} \ByDef$ \= $\rec x.\;{\tt Login}.($\=$\Dual{\tt Ok}.\Dual{\tt Id}.($\= \kill

${\sf BallotServiceBehP} \ByDef \rec x.\;{\tt Login}.(\Dual{\tt Wrong}.x$\\
                                        \>                  \> $\;\oplus\; $\\
                                        \>  \> $\Dual{\tt Overload}.x $ \\
                                         \> \> $\;\oplus\; $\\
                                         \>\> $\Dual{\tt Ok}.\Dual{\tt Id}.($\=${\tt VoteA}.({\tt Va1} + {\tt Va2})$\\
             \>                \>                  \>$+$\\
             \>                \>                  \>${\tt VoteB}.({\tt Vb1} + {\tt Vb2})$ \\
             \>                \>                  \>$+$\\
             \>                \>                  \>${\tt VoteC}.({\tt Vc1} + {\tt Vc2})) $
                                                     $)$
\end{tabbing}
}
Now, let us assume to have a voter with the following behaviour:
\[{\sf VoterBehP} \,\ByDef\,  \rec x.\;\Dual{\tt Login}.({\tt Wrong}.\,x + {\tt Overload}.\,x + {\tt Ok}.\Dual{\tt Vb1}.\Dual{\tt VoteB} ) \]
Such a voter, besides not needing any identifier of 
the transaction, intends to give the preference for the vice-candidate {\em before} the one for the main candidate.
The feasibility of the interaction between ${\sf VoterBehP}$ and ${\sf BallotServiceBehP}$ can be guaranteed only by the presence of an orchestrator such as:
\begin{tabbing}
$\mbox{{\sf BallotOrchP}} \;\ByDef\;\rec x.\;\langle{\tt Login},\Dual{\tt Login} \rangle.($ \=\kill

$\mbox{{\sf BallotOrchP}} \;\ByDef \; \rec x.\;\langle{\tt Login},\Dual{\tt Login} \rangle.(\;\;\;\;\;\langle\Dual{\tt Wrong},{\tt Wrong} \rangle.\,x $\\
       \> $\vee \; \langle\Dual{\tt Overload},{\tt Overload} \rangle.\,x $\\
        \> $\vee \; \langle\Dual{\tt Ok},{\tt Ok} \rangle.\langle\varepsilon, {\tt Id} \rangle.\langle{\tt Vb1},\varepsilon \rangle.\langle {\tt VoteB},\Dual{\tt VoteB}\rangle.\langle\varepsilon, \Dual{\tt Vb1}\rangle )$
\end{tabbing} 
\Comment{
\begin{tabbing}
$f \ByDef\langle{\tt Login},\Dual{\tt Login} \rangle.($ \=\kill

$f \ByDef\langle{\tt Login},\Dual{\tt Login} \rangle.(\langle\Dual{\tt Wrong},{\tt Wrong} \rangle $\\
       \> $\vee $\\
       \> $\langle\Dual{\tt Overload},{\tt Overload} \rangle $\\
        \> $\vee $\\
        \> $\langle\Dual{\tt Ok},{\tt Ok} \rangle.\langle\varepsilon, {\tt Id} \rangle.\langle{\tt VoteB},\varepsilon \rangle.\langle {\tt Vb1},\Dual{\tt Vb1},\rangle.\langle\varepsilon, \Dual{\tt VoteB},\rangle )$
\end{tabbing} 
}
The actions of an orchestrator are actually pairs.
The first orchestrating action  $\langle{\tt Login},\Dual{\tt Login}\rangle$
means that  {\sf BallotOrch} immediately delivers to the server a login, represented by the action $\Dual{\tt Login}$ to the right of the first pair,  as soon as this is received from the client, represented by the action {\tt Login} to the left of the same pair. Then,
the orchestrating actions  $\langle\Dual{\tt Wrong},{\tt Wrong} \rangle$, $\langle\Dual{\tt Overload},{\tt Overload} \rangle$ and $\langle\Dual{\tt Ok},{\tt Ok} \rangle$,
and the use of the $\vee$ operator, express that,
 in case {\sf BallotOrch} gets a message ${\tt Wrong}$, ${\tt Overload}$ or ${\tt Ok}$ from the server, this message is immediately passed to the client (and the orchestration starts again in case of ${\tt Wrong}$ or ${\tt Overload}$).
 
In case the message  ${\tt Ok}$ is received, the subsequent orchestrating actions begin by $\langle\varepsilon, {\tt Id} \rangle.\langle{\tt VoteB},\varepsilon \rangle$. The symbol $\varepsilon$ represents a  {\em no-action} by the client and by the server respectively, and it has the effect of buffering the other action in the pair.
Therefore the message ${\tt Id}$ from the server
is kept in a buffer since  the no-action symbol $\varepsilon$ to the left of the first orchestrating action replaces the expected $\Dual{\tt Id}$. Simlarly the message ${\tt Vb1}$ is also kept in the buffer. Only after the reception
of the message ${\tt VoteB}$, which is immediately passed to the server, the message
${\tt Vb1}$ is delivered to the server, and the orchestration stops. The message ${\tt Id}$,  instead, is never delivered.

The presence of an orchestrator hence allows for both asynchronous interactions and  the possibility of disregarding messages. 
A natural restriction is  imposed on orchestrators in \cite{Padovani10}:  an orchestrator cannot send a message if this has not been previously received. In fact, in the correct orchestrator  above, $\langle{\tt VoteB},\varepsilon \rangle$ comes before $\langle\varepsilon, \Dual{\tt VoteB}\rangle$.
This  implies that also in Padovani's setting it is not possible to disregard input actions.

\Comment{ - somewhat 
resembles the restriction imposed on the syntax of session-behaviours. 
For instance, we have that $\Dual{a}.b \oplus \Dual{c}.d\not\preceqP
b.\Dual{a} \oplus \Dual{c}.d$. In fact, let us consider one of the possible clients of
 $\Dual{a}.b \oplus \Dual{c}.d$, namely $a.\Dual{b} + c.\Dual{d}$.
There can be no orchestrator between $a.\Dual{b} + c.\Dual{d}$ and  the server  $b.\Dual{a} \oplus \Dual{c}.d$, since no orchestrator is allowed to send a message $a$ to the client if that has not been received. Similarly, it is not possible to take into account $b.\Dual{a} \oplus \Dual{c}.d$ as a server in the session-behaviours formalism, since it is not a sintactically correct.)\\
}


The generality of Padovani's  notion of orchestrated compliance 
is paid in terms of a more complex LTS formalizing client/server interaction, which depends
on an orchestrator $f$,
that we denote by $\complyP_f$.
In  \cite{Padovani10} the relation $\rho \complyP \sigma$ holds whenever there exists an orchestrator $f$ such that $\rho \complyP_f \sigma$.

To save decidability of the relevant properties, any correct orchestrator  
must be of {\em finite} rank, where the rank of an orchestrator $f$ is the bound of its buffering capability. 
To make this explicit the notation $\rho \complyP_k \sigma$ is used
whenever there exists an orchestrator $f$ of rank $k$ such that $\rho \complyP_f \sigma$.

In \cite{Padovani10} the sub-behaviour relation induced by orchestrated compliance is defined by:
\[\sigma\preceqP\sigma' ~ \Iff  ~  \forall \rho.[\;\rho\comply\sigma \implies \exists f.\; \rho\complyP_f \sigma']. \]
Notice that the relation $\comply$ in the antecedent of the implication is just the usual strong compliance.
In the same work the relation $\preceqP$ is  proved to be decidable. Moreover the orchestrator $f$ in the definition can be inferred 
from $\sigma$ and $\sigma'$ and it is the same for any possible client $\rho$.

From what said up to now Padovani orchestrated-compliance relation seems to include ours,
since the possibility of {\em skipping} output actions can be mimicked by  orchestrators
that keep messages indefinitely inside their buffers, without ever delivering them.

However, apart from the restriction to session behaviours,
the two compliance relations are actually {\em incomparable} because of the finiteness of the ranks of correct orchestrators and of the possibility in our setting to discard infinitely many (non consecutive) output actions from the server side. 
A counterexample to the inclusion of $\complyG$ in $\complyP$ can be obtained by slightly modifying the example used before.
Let us consider the ballot service with the extra output action $\Dual{\tt InfoW}$,
representing some informations about why a login has not been accepted:
\begin{tabbing}
${\sf Ball} \ByDef$ \= \,$\rec x.\;{\tt Login}.(\Dual{\tt Wrong}.\Dual{\tt InfoW}.x\;\oplus\; \Dual{\tt Overload}.x \;\oplus\; \Dual{\tt Ok}.\Dual{\tt Id}.($\= \kill

${\sf BallotServiceBehP2} \ByDef $\\
   \> $\rec x.\;{\tt Login}.(\Dual{\tt Wrong}.\Dual{\tt InfoW}.\,x\;\oplus\; \Dual{\tt Overload}.\,x \;\oplus\; \Dual{\tt Ok}.\Dual{\tt Id}.(\;\;\;\;{\tt VoteA}.({\tt Va1} + {\tt Va2})$\\
             \>                                  \>$+\; {\tt VoteB}.({\tt Vb1} + {\tt Vb2})\;)\;)$
\end{tabbing}
\Comment{
\begin{tabbing}
${\sf BallotServiceBehP2} \ByDef$ \= $\rec x.\;{\tt Login}.($\=$\Dual{\tt Ok}.\Dual{\tt Id}.($\= \kill

${\sf BallotServiceBehP2} \ByDef \rec x.\;{\tt Login}.(\Dual{\tt Wrong}.\Dual{\tt InfoW}.x$\\
                                        \>                  \> $\;\oplus\; $\\
                                        \>  \> $\Dual{\tt Overload}.x $ \\
                                         \> \> $\;\oplus\; $\\
                                         \>\> $\Dual{\tt Ok}.\Dual{\tt Id}.($\=${\tt VoteA}.({\tt Va1} + {\tt Va2})$\\
             \>                \>                  \>$+$\\
             \>                \>                  \>${\tt VoteB}.({\tt Vb1} + {\tt Vb2})$ \\
             \>                \>                  \>$+$\\
             \>                \>                  \>${\tt VoteC}.({\tt Vc1} + {\tt Vc2})) $
                                                     $)$
\end{tabbing}
}
Consider now the behaviour of a possible voter who can indefinitely try to log-in until (if ever) the login is accepted. This voter is not interested about why
a login has not been accepted, nor it is interested in getting the transaction identifier. Also it does not wish to express a vote for the vice-candidates:\\
${\sf VoterBehP2} = \rec  x.\;\Dual{\tt Login}.({\tt Wrong}.\;x + {\tt Overload}.\;x + {\tt Ok}.\Dual{\tt VoteB} ).$\\
Then we have that ${\sf VoterBehP2} \not\complyP {\sf BallotServiceBehP2}$, that is ${\sf VoterBehP2}$ is not compliant with the server ${\sf BallotServiceBehP2}$ according to  
the Padovani's orchestrated compliance. In fact the voter could keep on sending an incorrect login indefinitely, but no correct orchestrator is allowed 
to buffer an unbounded number of messages, like the ${\tt InfoW}$ ones.
As a matter of fact, the actual interaction between ${\sf VoterBehP2}$ and ${\sf BallotServiceBehP2}$ should be carried on in Padovani's setting through the use of an orchestrator like the following one:
\begin{tabbing}
${\sf BallotOrchP2} \;\ByDef \;$ \= $\rec x.\;\orchAct{{\tt Login}}{{\tt Login}}.$\=$\Dual{\tt Ok}.\Dual{\tt Id}$\= \kill

${\sf BallotOrchP2} \;\ByDef \; \rec x.\;\orchAct{{\tt Login}}{\Dual{\tt Login}}.(\;\;\;\;\orchAct{\Dual{\tt Wrong} }{{\tt Wrong}}.\orchAct{\varepsilon }{{\tt InfoW}}.x$\\
                                        \>                  \> $\;\vee\; \orchAct{\Dual{\tt Overload} }{{\tt Overload}}.~x $ \\
                                         \> \> $\;\vee\; \orchAct{\Dual{\tt Ok} }{{\tt Ok}}.\orchAct{\varepsilon }{{\tt Id}}.$\=$\orchAct{{\tt VoteB}}{\Dual{\tt VoteB}} \;)$
\end{tabbing}
\Comment{
\begin{tabbing}
${\sf BallotOrch2} =$ \= $\rec x.\;\orchAct{{\tt Login}}{{\tt Login}}.($\=$\Dual{\tt Ok}.\Dual{\tt Id}.($\= \kill

${\sft BallotOrch2} = \rec x.\;\orchAct{{\tt Login}}{\Dual{\tt Login}}.(\orchAct{\Dual{\tt Wrong} }{{\tt Wrong}}.\orchAct{\varepsilon }{{\tt InfoW}}.x$\\
                                        \>                  \> $\;\vee\; $\\
                                        \>  \> $\orchAct{\Dual{\tt Overload} }{{\tt Overload}}.~x $ \\
                                         \> \> $\;\vee\; $\\
                                         \>\> $\orchAct{\Dual{\tt Ok} }{{\tt Ok}}.\orchAct{\varepsilon }{{\tt Id}}.($\=$\orchAct{{\tt VoteB}}{\Dual{\tt VoteB}}$
                                                     $)$
\end{tabbing}
}
which should be able to buffer an unbounded number of ${\tt InfoW}$ messages corresponding to the output actions $\Dual{\tt InfoW}$ on the server side. This implies that {\sf BallotOrchP2} is 
 not of finite rank and hence it is not correct.
 
The definition of $\skipAct$-compliance allows to diregard infinitely many output actions from the server, provided that they are
not all consecutive. In particular
${\sf VoterBehP2} \complyG {\sf BallotServiceBehP2}$. So, formally we get:

\begin{proposition} 
\label{prop:counterPadovani}
Let $\complyP$ be Padovani's weak $k$-compliance restricted to session behaviours.
Then, for any $k$, we have:
\vspace{-3mm} 
$$\complyG ~\not\subseteq ~\complyP_k$$
\end{proposition}
\noindent

The inclusion does hold, instead, if we consider only {\em finite} behaviours (see Definition \ref{def:traces}):

\begin{proposition} For any pair of {\em finite} session behaviours $\rho,\sigma$, there exists a $k\geq0$
such that 
\vspace{-1mm}
$$\rho\complyG\sigma ~\implies ~\rho\complyP_k\sigma$$
\end{proposition}

In a sense, we think that the $\skipAct$-compliance relation we investigate in the present paper is
the minimal weakening of the standard notion of compliance not requiring the introduction of  orchestrators.

\section{Conclusion and future work}
\label{sect:futureWork}
\newcommand{\complyR}{\comply^{\blacktriangle}}
\newcommand{\ckpt}{_\blacktriangle\!}

In the setting of session-behaviors we have relaxed the synchronization rules by allowing output actions on the server side to be skipped by a client that cannot immediately synchronize with them. This gives rise to a weaker notion of compliance, called $\skipAct$-compliance, and consequently to a new concept of sub-behaviour among servers. We have proved that $\skipAct$-compliance is still decidable, by exhibiting a derivation system which is sound and complete w.r.t. the new compliance relation, and which is algorithmic, namely it implicitly describes an algorithm to decide $\skipAct$-compliance. Further we have shown that the duals-as-minima property is preserved in the new setting, which implies decidability of the induced sub-behaviour relation.

In the Introduction we have justified the loosening of compliance by means of examples. Another contexts 
in which discarding some actions during client/server interaction seems a desirable feature worth to be investigated is that 
of reversible computations. In particular when the client (or server) of an interaction can roll-back to a  previously encountered checkpoint (so forcing a roll-back on the server (client) side).
Then the notion of compliance should be strengthened to guarantee that client's requests keep on being satisfied even in case, for any reason, client and server perform a roll-back, as formalized and investigated in \cite{BDdL14}. It is not difficult to envisage a situation where the interaction partners could roll-back in two states that would be compliant but for the presence of an output that should have been already sent and received before the roll-back took place. It is reasonable to let the two partners be compliant, since that particular output action could be safely discarded.\\
{\bf Acknowledgements.} The authors wish to thank Mariangiola Dezani for her steady and valuable support. Our gratitude also to the anonymous referees that helped us to improve
the paper.
\vspace{-5mm}





\bibliographystyle{eptcs}

\begin{thebibliography}{10}
\providecommand{\bibitemdeclare}[2]{}
\providecommand{\surnamestart}{}
\providecommand{\surnameend}{}
\providecommand{\urlprefix}{Available at }
\providecommand{\url}[1]{\texttt{#1}}
\providecommand{\href}[2]{\texttt{#2}}
\providecommand{\urlalt}[2]{\href{#1}{#2}}
\providecommand{\doi}[1]{doi:\urlalt{http://dx.doi.org/#1}{#1}}
\providecommand{\bibinfo}[2]{#2}

\bibitemdeclare{article}{BD91}
\bibitem{BD91}
\bibinfo{author}{Eric \surnamestart Badouel\surnameend} \&
  \bibinfo{author}{Philippe \surnamestart Darondeau\surnameend}
  (\bibinfo{year}{1991}): \emph{\bibinfo{title}{On Guarded Recursion.}}
\newblock {\sl \bibinfo{journal}{Theor. Comput. Sci.}}
  \bibinfo{volume}{82}(\bibinfo{number}{2}), pp. \bibinfo{pages}{403--408},
  \doi{10.1016/0304-3975(91)90231-P}.

\bibitemdeclare{inproceedings}{BdL10}
\bibitem{BdL10}
\bibinfo{author}{Franco \surnamestart Barbanera\surnameend} \&
  \bibinfo{author}{Ugo \surnamestart de'Liguoro\surnameend}
  (\bibinfo{year}{2010}): \emph{\bibinfo{title}{Two notions of sub-behaviour
  for session-based client/server systems}}.
\newblock In: {\sl \bibinfo{booktitle}{Proceedings of PPDP'10}},
  \bibinfo{publisher}{ACM}, pp. \bibinfo{pages}{155--164}
  \doi{10.1145/1836089.1836109}.

\bibitemdeclare{inproceedings}{BDdL14}
\bibitem{BDdL14}
\bibinfo{author}{Franco \surnamestart Barbanera\surnameend},
  \bibinfo{author}{Mariangiola \surnamestart Dezani\surnameend} \&
  \bibinfo{author}{Ugo \surnamestart de' Liguoro\surnameend}
  (\bibinfo{year}{2014}): \emph{\bibinfo{title}{Compliance for reversible
  client/server interactions}}.
\newblock In: {\sl \bibinfo{booktitle}{Proceedings of BEAT 2014}},
  \bibinfo{series}{EPTCS}.
\newblock \bibinfo{note}{To appear}.

\bibitemdeclare{article}{BdL13}
\bibitem{BdL13}
\bibinfo{author}{Franco \surnamestart Barbanera\surnameend} \&
  \bibinfo{author}{Ugo \surnamestart de' Liguoro\surnameend}
  (\bibinfo{year}{2013}): \emph{\bibinfo{title}{{Sub-behaviour relations for
  session-based client/server systems}}}.
\newblock {\sl \bibinfo{journal}{MSCS}}.
\newblock \bibinfo{note}{To appear}.

\bibitemdeclare{article}{BH13}
\bibitem{BH13}
\bibinfo{author}{Giovanni \surnamestart Bernardi\surnameend} \&
  \bibinfo{author}{Matthew \surnamestart Hennessy\surnameend}
  (\bibinfo{year}{2013}): \emph{\bibinfo{title}{Modelling session types using
  contracts}}.
\newblock {\sl \bibinfo{journal}{To appear in Mathematical Structures in
  Computer Science}}.

\bibitemdeclare{inproceedings}{CCLP06}
\bibitem{CCLP06}
\bibinfo{author}{S.~\surnamestart Carpineti\surnameend},
  \bibinfo{author}{G.~\surnamestart Castagna\surnameend},
  \bibinfo{author}{C.~\surnamestart Laneve\surnameend} \&
  \bibinfo{author}{L.~\surnamestart Padovani\surnameend}
  (\bibinfo{year}{2006}): \emph{\bibinfo{title}{A formal account of contracts
  for {Web} {S}ervices}}.
\newblock In: {\sl \bibinfo{booktitle}{WS-FM, 3rd Int.\ Workshop on Web
  Services and Formal Methods}}, {\sl \bibinfo{series}{LNCS}}
  \bibinfo{volume}{4184}, \bibinfo{publisher}{Springer}, pp.
  \bibinfo{pages}{148--162}, \doi{10.1007/11841197\_10}.

\bibitemdeclare{inproceedings}{CGP08}
\bibitem{CGP08}
\bibinfo{author}{G.~\surnamestart Castagna\surnameend},
  \bibinfo{author}{N.~\surnamestart Gesbert\surnameend} \&
  \bibinfo{author}{L.~\surnamestart Padovani\surnameend}
  (\bibinfo{year}{2008}): \emph{\bibinfo{title}{A Theory of Contracts for Web
  Services}}.
\newblock In: {\sl \bibinfo{booktitle}{POPL~'08, 35th ACM Symposium on
  Principles of Programming Languages}}, 
\doi{10.1145/1328438.1328471}.

\bibitemdeclare{inproceedings}{CP09}
\bibitem{CP09}
\bibinfo{author}{G.~\surnamestart Castagna\surnameend},
  \bibinfo{author}{N.~\surnamestart Gesbert\surnameend} \&
  \bibinfo{author}{L.~\surnamestart Padovani\surnameend}
  (\bibinfo{year}{2009}): \emph{\bibinfo{title}{Contracts for mobile
  processes}}.
\newblock In: {\sl \bibinfo{booktitle}{Proceedings of the 20th International
  Conference on Concurrency Theory (CONCUR'09)}}, {\sl \bibinfo{series}{LNCS}}
  \bibinfo{volume}{5710}, \bibinfo{publisher}{Springer}, pp.
  \bibinfo{pages}{211--228},
\doi{10.1145/1538917.1538920}.

\bibitemdeclare{article}{CGP10}
\bibitem{CGP10}
\bibinfo{author}{Giuseppe \surnamestart Castagna\surnameend},
  \bibinfo{author}{Nils \surnamestart Gesbert\surnameend} \&
  \bibinfo{author}{Luca \surnamestart Padovani\surnameend}
  (\bibinfo{year}{2009}): \emph{\bibinfo{title}{A theory of contracts for Web
  services}}.
\newblock {\sl \bibinfo{journal}{ACM Trans. Program. Lang. Syst.}}
  \bibinfo{volume}{31}(\bibinfo{number}{5}), pp. \bibinfo{pages}{19:1--19:61},
  \doi{10.1145/1538917.1538920}.

\bibitemdeclare{article}{GH05}
\bibitem{GH05}
\bibinfo{author}{Simon \surnamestart Gay\surnameend} \&
  \bibinfo{author}{Malcolm \surnamestart Hole\surnameend}
  (\bibinfo{year}{2005}): \emph{\bibinfo{title}{{Subtyping for Session Types in
  the Pi-Calculus}}}.
\newblock {\sl \bibinfo{journal}{Acta Informatica}}
  \bibinfo{volume}{42}(\bibinfo{number}{2/3}), pp. \bibinfo{pages}{191--225},
  \doi{10.1007/s00236-005-0177-z}.

\bibitemdeclare{inproceedings}{honda.vasconcelos.kubo:language-primitives}
\bibitem{honda.vasconcelos.kubo:language-primitives}
\bibinfo{author}{Kohei \surnamestart Honda\surnameend},
  \bibinfo{author}{Vasco~T. \surnamestart Vasconcelos\surnameend} \&
  \bibinfo{author}{Makoto \surnamestart Kubo\surnameend}
  (\bibinfo{year}{1998}): \emph{\bibinfo{title}{{Language Primitives and Type
  Disciplines for Structured Communication-based Programming}}}.
\newblock In: {\sl \bibinfo{booktitle}{ESOP'98}}, {\sl \bibinfo{series}{LNCS}}
  \bibinfo{volume}{1381}, \bibinfo{publisher}{Springer-Verlag}, pp.
  \bibinfo{pages}{22--138},
   \doi{10.1007/BFb0053567}.

\bibitemdeclare{inproceedings}{LP07}
\bibitem{LP07}
\bibinfo{author}{Cosimo \surnamestart Laneve\surnameend} \&
  \bibinfo{author}{Luca \surnamestart Padovani\surnameend}
  (\bibinfo{year}{2007}): \emph{\bibinfo{title}{{The Must Preorder Revisited:
  An Algebraic Theory for Web Services Contracts}}}.
\newblock In: {\sl \bibinfo{booktitle}{CONCUR'07}}, {\sl
  \bibinfo{series}{LNCS}} \bibinfo{volume}{4703},
  \bibinfo{publisher}{Springer-Verlag}, pp. \bibinfo{pages}{212--225},
  \doi{10.1007/978-3-540-74407-8\_15}.

\bibitemdeclare{inproceedings}{DeNicolaH83}
\bibitem{DeNicolaH83}
\bibinfo{author}{Rocco~De \surnamestart Nicola\surnameend} \&
  \bibinfo{author}{Matthew \surnamestart Hennessy\surnameend}
  (\bibinfo{year}{1983}): \emph{\bibinfo{title}{Testing Equivalence for
  Processes}}.
\newblock In: {\sl \bibinfo{booktitle}{ICALP}}, {\sl \bibinfo{series}{LNCS}}
  \bibinfo{volume}{154}, \bibinfo{publisher}{Springer}, pp.
  \bibinfo{pages}{548--560},
  \doi{10.1007/BFb0036936}.

\bibitemdeclare{article}{Padovani10}
\bibitem{Padovani10}
\bibinfo{author}{Luca \surnamestart Padovani\surnameend}
  (\bibinfo{year}{2010}): \emph{\bibinfo{title}{{C}ontract-{B}ased {D}iscovery
  of {W}eb {S}ervices {M}odulo {S}imple {O}rchestrators}}.
\newblock {\sl \bibinfo{journal}{Theoretical Computer Science}}
  \bibinfo{volume}{411}, pp. \bibinfo{pages}{3328--3347},
  \doi{10.1016/j.tcs.2010.05.002}.

\bibitemdeclare{inproceedings}{PierceSangiorgi95}
\bibitem{PierceSangiorgi95}
\bibinfo{author}{Benjamin~C. \surnamestart Pierce\surnameend} \&
  \bibinfo{author}{Davide \surnamestart Sangiorgi\surnameend}
  (\bibinfo{year}{1996}): \emph{\bibinfo{title}{Typing and Subtyping for Mobile
  Processes}}.
\newblock {\sl \bibinfo{journal}{Mathematical Structures in Computer Science}}
  \bibinfo{volume}{6, No.\ 5}.

\bibitemdeclare{article}{Sangio09}
\bibitem{Sangio09}
\bibinfo{author}{Davide \surnamestart Sangiorgi\surnameend}
  (\bibinfo{year}{2009}): \emph{\bibinfo{title}{On the origins of bisimulation
  and coinduction}}.
\newblock {\sl \bibinfo{journal}{ACM Trans. Program. Lang. Syst.}}
  \bibinfo{volume}{31}(\bibinfo{number}{4}), \doi{10.1145/1516507.1516510}.

\bibitemdeclare{article}{Tar55}
\bibitem{Tar55}
\bibinfo{author}{Alfred \surnamestart Tarski\surnameend}
  (\bibinfo{year}{1955}): \emph{\bibinfo{title}{A lattice-theoretical fixpoint
  theorem and its applications}}.
\newblock {\sl \bibinfo{journal}{Pac. J. Math.}} \bibinfo{volume}{5}, p.
  \bibinfo{pages}{285–309}, \doi{10.2140/pjm.1955.5.285}.

\end{thebibliography}

\end{document}